\documentclass[10pt,journal,compsoc]{IEEEtran}
\usepackage{amsmath, amsthm, amssymb}
\usepackage{algorithm}
\usepackage[noend]{algpseudocode}
\usepackage{romannum}

\usepackage{float}
\usepackage{subfig}
\usepackage{graphicx}
\usepackage{longtable}
\usepackage{multirow}

\usepackage{array}
\newcolumntype{C}[1]{>{\centering\arraybackslash}m{#1}}

\newtheorem{theorem}{\textbf{Theorem}}
\newtheorem*{process*}{\textbf{Process}}

\newtheorem{lemma}{\textbf{Lemma}}
\newtheorem{claim}{\textbf{Claim}}

\newtheorem{corollary}{\textbf{Corollary}}

\theoremstyle{definition}

\newtheorem{problem}{\textbf{Problem}}
\newtheorem{definition}{\textbf{Definition}}
\DeclareMathOperator{\T}{\mathcal{T}}
\DeclareMathOperator*{\argmax}{arg\,max}
\DeclareMathOperator*{\argmin}{arg\,min}
\DeclareMathOperator{\dis}{dis}

\ifCLASSOPTIONcompsoc
  \usepackage[nocompress]{cite}
\else
  \usepackage{cite}
\fi
\ifCLASSINFOpdf
\else
\fi
\begin{document}
\title{An Efficient Randomized Algorithm for Rumor Blocking in Online Social Networks}

\author{
	Guangmo (Amo)~Tong,~\IEEEmembership{Student Member,~IEEE,}
	Weili~Wu,~\IEEEmembership{Member,~IEEE,}
	Ling~Guo,
	Deying~Li,
	Cong~Liu,~\IEEEmembership{Member,~IEEE,}
	Bin~Liu,
	and~Ding-Zhu~Du% <-this % stops a space
	\IEEEcompsocitemizethanks{
		\IEEEcompsocthanksitem G. Tong, W. Wu, C. Liu and D.-Z. Du are with the Department of Computer Science Erik
		Jonsson School of Engineering and Computer Science The University
		of Texas at Dallas, 800 W. Campbell Road; MS EC31 Richardson, TX
		75080 U.S.A.\protect\\
		E-mail: guangmo.tong@utdallas.edu
		\IEEEcompsocthanksitem L. Guo and D. Li are with the Rennin University.
		\IEEEcompsocthanksitem B. Liu is with the Ocean University of China.}% <-this % stops an unwanted space
\thanks{Manuscript received April 19, 2005; revised August 26, 2015.}
}


\markboth{Journal of \LaTeX\ Class Files,~Vol.~14, No.~8, August~2015}%
{Shell \MakeLowercase{\textit{et al.}}: Bare Demo of IEEEtran.cls for Computer Society Journals}



\IEEEtitleabstractindextext{%
\begin{abstract}
Social networks allow rapid spread of ideas and innovations while negative information can also propagate widely. When a user receives two opposing opinions, they tend to believe the one arrives first. Therefore, once misinformation or rumor is detected, one containment method is to introduce a positive cascade competing against the rumor. Given a budget $k$, the rumor blocking problem asks for $k$ seed users to trigger the spread of a positive cascade such that the number of the users who are not influenced by rumor can be maximized. The prior works have shown that the rumor blocking problem can be approximated within a factor of $(1-1/e)$ by a classic greedy algorithm combined with Monte Carlo simulation. Unfortunately, the Monte Carlo simulation based methods are time consuming and the existing algorithms either trade performance guarantees for practical efficiency or vice versa. In this paper, we present a randomized approximation algorithm which is provably superior to the state-of-the art methods with respect to running time. The superiority of the proposed algorithm is demonstrated by experiments done on both the real-world and synthetic social networks. 
\end{abstract}

}

\maketitle

\IEEEdisplaynontitleabstractindextext

% \fi
\IEEEpeerreviewmaketitle

\IEEEraisesectionheading{\section{Introduction}\label{sec:introduction}}

\IEEEPARstart{T}he tremendous advance of the Internet of things (loT) is making online social networks be the most common platform for communication. There have been totally 44.5 million users on Twitter and 1.4 million monthly active users on Facebook. Admittedly online social networks are greatly beneficial, they also lead the widespread of negative information.  Such negative influence, namely misinformation and rumor, has been a cause of concern as it renders the network unreliable and may cause further panic in population. For example, the misinformation of swine flu on Twitter threw the people in Texas and Kansas into panic in 2009 \cite{morozov2009swine}, and the endless report of Ebola in 2014 has caused unnecessary worldwide terror.  Therefore, effective strategies for rumor containment are crucial for social networks and it has been a hot topic in the last decades.

In a social network, information and innovations diffuse from user to user via influence cascades where each cascade starts to spread with a set of seed users. When two cascades holding opposing views reach a certain user, the user is likely to trust the one arriving first. As an example, if the international institutions like WHO would have posted clarification for swine flu, the users who have read such posts will not be misled by the misinformation. Therefore, a common method for rumor blocking is to generate a corresponding positive cascade that competes against the rumor. Due to the expense of deploying seed nodes, there is a budget $k$ for the positive cascade, and naturally one selects the $k$ positive seed nodes which are able to limit the spread of rumor in maximum, which is referred as the \textit{least cost rumor blocking problem}. 

The recent study of influence diffusion in social networks can be tracked back to D. Kempe \cite{kempe2003maximizing} where the well-known influence maximization problem is formulated. In that seminal work, two fundamental probabilistic operational models, independent cascade (IC) model and linear threshold (LT) model are developed. Based on such models, many influence related problems are then proposed and studied. The problem of rumor blocking is also considered in such models or in their variants. Most existing approaches utilize the submodularity of the objective function. A set function $f$ over a ground $V$ is set to be submodular if 
\begin{equation}
\label{eq: submodular}
f(V_1 \cup \{v\})-f(V_1) \geq f(V_2 \cup \{v\})-f(V_2)
\end{equation}
for any $V_1 \subset V_2 \subseteq V$ and $v \in V \setminus V_2$. It turns out that the number of non-rumor-activated users is a monotone increasing submodular function and consequently the classic hill-climbing algorithm provides a $(1-1/e)$-approximation \cite{nemhauser1978analysis}. For example, X. He \textit{et al.} \cite{he2012influence} formulate the influence blocking maximization problem and present a $(1-1/e)$-approximation algorithm for the competitive linear threshold model, Budak \textit{et al.} \cite{budak2011limiting} propose several competitive models and show a greedy algorithm with the same approximation ratio under the campaign-oblivious independent cascade model, and, Fan \textit{et al.} \cite{fan2013least} provide a $(1-1/e)$-approximation algorithm for the rumor blocking problem under the opportunistic one-active-one model. 

Assuming that the objective function can be efficiently\footnote{Usually this is referred to the polynomial-time computability.} calculated for any input, the greedy algorithm is simple and effective for most of the submodular maximization problems. Unfortunately, for the influence related optimization problems, the objective functions are often very complicated to compute due to the randomness of the probabilistic diffusion model. Such a scenario is first observed by W. Chen \cite{chen2010scalable} where it is shown that computing the exact value of the expected influence is \#P-hard. In order to circumvent such difficulty, the prior works employ the Monte Carlo simulation to estimate the objective value for any input. However, such a method is computationally expensive. It turns out that the greedy algorithm with Monte Carlo simulation has the $\Omega(k\cdot m\cdot n \cdot \text{poly}(\delta^{-1}))$ time complexity to achieve a $(1-1/e-\delta)$ approximation ratio, and it takes several hours even on very small networks. With the recently analysis of influence diffusion \cite{borgs2014maximizing, tang2014influence, tang2015influence}, the difficulty in solving such problems has shifted from the nodes selection strategy to the calculation of the objective function. \textit{Fundamentally, it asks for a better sampling method to estimate the expected influence.} To the best of our knowledge, there is no rumor blocking algorithm that can meet practical efficiency without sacrificing performance guarantee. 

In this paper, we present an efficient randomized algorithm for rumor blocking, which is termed as the reverse-tuple (R-tuple) based randomized (RBR) algorithm. The RBR algorithm runs in $O(\frac{k \cdot m \cdot \ln n}{\delta^2})$ and returns a $(1-1/e-\delta)$ approximation with a provable probability. The proposed algorithm utilizes the R-tuple based sampling method which is more effective than the Monte Carlo simulation used in the prior works. The reverse sampling technique is first designed by C. Borgs \cite{borgs2014maximizing} for the influence maximization problem. In this paper, we develop a new type of sampling based on the concept of R-tuple, and show how such sampling method can be applied to the rumor blocking problem. Although both the sampling methods give the unbiased estimate, one set of Monte Carlo simulations can only provide an estimation for a specified seed set, while the samples produced by the R-tuple based sampling can be applied to any seed sets. The RBR algorithm can be implemented with tunable parameters and it is flexible for balancing the running time and the error probability. We experimentally evaluate the proposed algorithm on both the real-world social network and synthetic power-law networks. The experimental results show that the RBR algorithm not only produces high quality positive seed set but also takes much less time than the greedy algorithm with the Monte Carlo simulation does. In particular, when $\delta=0.1$ and the error probability is set as less than $1/n$ where $n$ is the number of users, RBR algorithm is at least 1,000 times faster than that of the sate-of-the-art approach. The contributions of this paper are summarized as follows:
\begin{itemize}
\item We develop the reverse-tuple based sampling method which can be used to obtain an unbiased estimate for the objective function of rumor blocking problem.

\item Based on the new sampling technique, we design the RBR approximation algorithm which is effective and efficient for blocking rumors under IC model.

\item We evaluate the proposed algorithm via experiments and show that the RBR algorithm outperforms the existing methods by a significant magnitude in terms of the running time.
\end{itemize}

The rest of the paper is organized as follows. Sec. \ref{sec:related} is devoted to the related work. The preliminaries are provided in Sec. \ref{sec:model}. The RBR algorithm is shown in Sec. \ref{sec:algorithm}. The experiments are presented in Sec. \ref{sec:exp}. In Sec. \ref{sec:con} we discuss the future work and conclude this paper. 

\section{Related Work}
\label{sec:related}
In this section, we survey the prior works regarding rumor controlling. 

C. Budak \textit{et al.} \cite{budak2011limiting} are among the first who study the misinformation containment problem. In particular, they consider the multi-campaign independent cascade model and investigate the problem of identifying a subset of individuals that need to be convinced to adopt the ``good" campaign so as to minimize the number of people that adopt the rumor. X. He \textit{et al.} \cite{he2012influence} and L. Fan  \textit{et la.} \cite{fan2013least} further study this problem under the competitive linear threshold model and the OPOAO model, respectively. S. Li \textit{et al.} \cite{li2013rumor} later formulate the $\gamma-k$ rumor restriction problem and show a $(1-1/e)$-approximation. As mentioned earlier, the existing approaches are time consuming and thus cannot handle large social networks. Recently, several heuristic methods have been proposed by different works, such as  \cite{zhang2015limiting,ping2014sybil}, but they cannot provide performance guarantee. In this paper, we aim to design the rumor blocking algorithm which is provably
 effective and also efficient. 

Rumor source detection is another important problem for rumor controlling. The prior works primarily focus on the susceptible-infected-recovered (SIR) model where the nodes can be infected by rumor and may recover later. Shah \textit{et al.} \cite{shah2011rumors} provide a systematic study and design a rumor source estimator based upon the concept of rumor centrality. Z. Wang \textit{et al.} \cite{wang2014rumor} propose a unified inference framework based on the union rumor centrality. 

Rumor detection aims to distinguish rumor from genuine news. Leskovec \textit{et al.} \cite{leskovec2009meme} develop a framework for tracking the spread of misinformation and observe a set of persistent temporal patterns in the news cycle. Ratkiewicz \textit{et al. }\cite{ratkiewicz2011detecting} build a machine learning framework to  detect the early stages of viral spreading of political misinformation. In \cite{qazvinian2011rumor}, Qazvinian \textit{et al.} address the rumor detection problem by exploring the effectiveness of three categories of features: content-based, network-based, and microblog-specific memes. Takahashi \textit{et al.} \cite{takahashi2012rumor} study the characteristics of rumor and design a system to detect rumor on Twitter. 

\section{System Model}
\label{sec:model}
In what follows we provide the preliminaries to the rest of this paper.  The important notations are listed in Table \ref{table:symbol}.

\subsection{Influence Model}
A social network is represented by a directed graph $G=(V,E)$ where $V$ denotes the user set and user $v$ is a friend of $u$ iff $(u,v)\in E$. Let $n$ and $m$ be the number of nodes and edges, respectively. We denote by $N_v^-$ the set of in-neighbors of $v$. For a network $G$, let $E(G)$ and $V(G)$ be the edge-set and node-set of $G$, respectively. In order to spread an idea or to advertise a new product in a social network, some seed nodes are chosen to be activated to trigger the spread of influence. The diffusion process terminates when there is no user can be further activated. In this paper, we adopt the following model.

%\textbf{Linear Threshold (LT) model.} 
%Associated with each edge $(u,v)$ there is an weight $w_{(u,v)} \in [0,1]$ and each node $u$ has a threshold $\theta_u$, where $\sum_v w_{(v,u)} \leq 1$. For a node $u$ other than the seed nodes, $u$ becomes active at time step $t$ if $\sum_{v \in A_{t-1}} w_{(v,u)} \geq \theta_u$, where $A_{t-1}$ is the set of the users activated at time $t-1$.

\textbf{Independent Cascade (IC) Model.}
Associated with each edge $(u,v)$ there is a propagation probability $p_{(u,v)}$. When node $u$ becomes active at time $t-1$, it attempts to activate each inactive neighbor $v$ at time step $t$ with a success probability of $p_{(u,v)}$. For each pair of nodes $u$ and $v$, $u$ has only one chance to activate $v$.

\begin{table}[t]
\centering
{\begin{tabular}{ |p{2.1cm} || p{5.cm} |}
\hline 
\textbf{Symbol}& \textbf{Definition} \\
\hline 
$G=(V,E)$& instance of IC network.    \\
\hline
$n$ and $m$&number of nodes and edges.  \\  
\hline
$N_v^-$& the set of in-neighbors of $v$\\
\hline
$p_e$& propagation probability of edge $e$.  \\  
\hline
$g$& a realization.  \\  
\hline
$\Pr[g]$& the probability that $g$ can be generated.  \\  
\hline
$S_r$ & the seed set of rumor.  \\  
\hline
$S_p$ & the seed set of positive cascade. \\  
\hline
$f(S_p)$& objective function of rumor blocking.  \\  
\hline
$OPT_i$& $OPT_i=f(S_i)$,\newline where $S_i=\argmax_{|S| \leq i}f(S)$.  \\  
\hline
$k$& the budget of the seed set of positive cascade.  \\  
\hline
$\T_v$& a random R-tuple of $v$.  \\  
\hline
$\Re_v$& the set of all possible random R-tuples of $v$.  \\  
\hline
$T_v$& a concrete R-tuple of $v$ in $\Re_v$.  \\  
\hline
$\Pr[T_v]$& the probability that $T_v$ can be generated by Alg. \ref{alg:r_tuple_v}.  \\  
\hline
$\T$& a random R-tuple.  \\  
\hline
\end{tabular}}
\caption{\textbf{Notations.} }
\label{table:symbol}
%\label{table:state1}
%\caption{States of edge $(v_{i,j},v_{0,i})$, for $i>0$.}
\end{table}

\subsection{Rumor and Competing Cascade}
Note that the IC model is originally designed for single cascade diffusion. Suppose there are multiple cascades each of which is generated by its own seed set. In the network, each node is initially inactive and never changes its state once activated by one cascade. Therefore, the cascade arriving first will dominate the node. In order to limit the spread of rumor, we introduce a competing cascade denoted as the positive cascade. At each time step, if a node is successfully activated by two or more neighbors belonging to different cascades, it will select the one with the highest priority. We assume that rumor has the higher priority, because rumor always polishes itself to be convincing. We denote by $S_r$ and $S_p$ the seed sets of rumor and the competing positive cascade, respectively. The diffusion process unfolds in discrete, as follows.
\begin{itemize}
\item Initially all the nodes are inactive.
\item At time 0, nodes in $S_r$ and $S_p$ are activated by rumor and positive cascade, respectively\footnote{Since our goal is to limit the spread of rumor and rumor has the higher priority, we can assume that $S_p \cap S_r=\emptyset$ without loss of generality.}.
\item At time $t>0$, each node $u$ which is activated at $t-1$ attempts to activate each of its inactive neighbors $v$ with a success probability of $p_{(u,v)}$. If node $v$ is successfully activated by the two cascades simultaneously at time $t$, $v$ will be activated by rumor. 
\item The diffusion process terminates when there is no node can be further activated.
\end{itemize}

An example is shown in Fig. \ref{fig: example}, where there are five nodes and the propagation probability of each edge is 1. In this example, $v_4$ and $v_3$ are selected as the seed node of rumor and positive cascade, respectively. At time step 1, $v_3$ and $v_4$ activate $v_5$ simultaneously and $v_5$ is finally activated by rumor as rumor has the higher priority.
\begin{figure}[t]
\begin{center}
\includegraphics[width=0.48\textwidth]{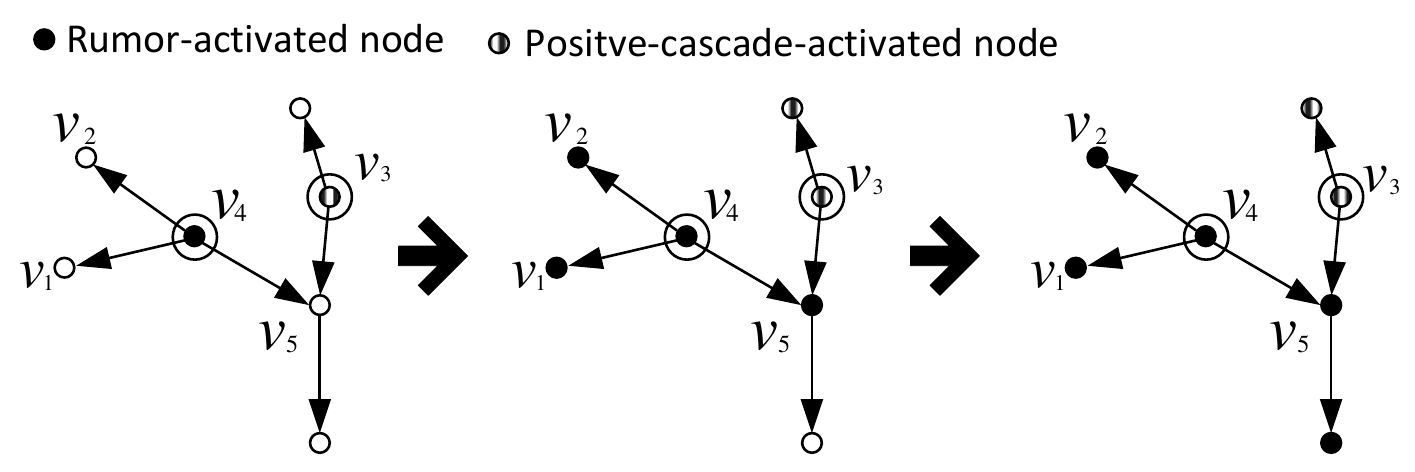} 
\end{center} 
\vspace{-5mm}
\caption{An illustrative example.}
\label{fig: example}
\end{figure}

\subsection{The problem}
Given an IC network $G$ and the seed set $S_r$ of rumor, let $f(S_p)$ be the expected number of nodes that are not activated by rumor when $S_p$ is selected as the seed set of the positive cascade. Given a budget $k \in \mathbb{Z}^+$, the rumor blocking problem considered in this paper is given as follows.
\begin{problem}
\label{problem:1}
Find a seed set $S_p$ with at most $k$ nodes such that $f(S_p)$ is maximized. 
\end{problem}
It is well-known that this problem is NP-hard.
\begin{theorem}\cite{budak2011limiting}
\label{theorem:nphard}
Problem \ref{problem:1} is NP-hard.
\end{theorem}
For $1 \leq i \leq k$, let $S_i=\argmax_{|S| \leq i}f(S)$ and $OPT_i=f(S_i)$. Since polynomial exact algorithm does not exist unless NP=P, we aim to design approximation algorithms.  
\subsection{Realization}
In this section, we introduce the concept of \textit{realization} which provides a fundamental understanding of IC model.
%\begin{definition}
%a realization $g$ is a three tuple $(g_1, g_2, g_3)$ where $g_1$, $g_2$ and $g_3$ form a partition of $E(G)$.
%\end{definition}
\begin{definition}
\label{def:realization}
Given an IC network $G$, a realization $g$ of $G$ is an IC network where  $V(g)=V(G)$ and $E(g)$ is a subset of $E(G)$ where each edge in $E(g)$ has the propagation probability of $1$. The edge set $E(g)$ is constructed in random. For each edge $e$ in $G$, we generate a random number $rand_e$ from 0 to 1 in uniform. Edge $e$ appears in $g$ if and only if $rand_e \leq p_e$. Let $\mathcal{G}$ be the set of all possible realizations of $G$. One can see that there are $2^{|E(G)|}$ realizations in $\mathcal{G}$.
\end{definition}

Let $\mathrm{Pr}[g]$ be the probability that realization $g$ can be generated. By Def. \ref{def:realization}, 
\begin{equation*}
\mathrm{Pr}[g]=\prod_{e \in E(g)}p_e\prod_{e \in E(G)\setminus E(g)}(1-p_e).
\end{equation*}
Intuitively, a realization $g$ is a deterministic IC network. Given the seed sets $S_r$ and $S_p$, the following two diffusion processes are equivalent with respect to $f(S_p)$ \cite{kempe2003maximizing}.
\begin{itemize}
\item Execute the stochastic diffusion process on $G$ with $S_r$ and $S_p$.
\item Randomly generate a realization $g$ of $G$ and execute the deterministic diffusion process on $g$ with $S_r$ and $S_p$.
\end{itemize}
Define that 
$$f_g(S_p,u) =
  \begin{cases}
  0 &  \hspace{0mm} \hspace{-0.5mm} \text{if $u$ is activated by rumor in $g$ under $S_p$} \\
  1 & \hspace{0mm} \hspace{-0.5mm} \text{else } 
  \end{cases}$$
Therefore, $f(S_p)$ can be expressed as 
\begin{equation}
\label{eq:f(S_p)}
f(S_p)=\sum_{g \in \mathcal{G}} \sum_{u \in V} \mathrm{Pr}[g]  f_g(S_p,u).
\end{equation}
In order to maximize $f(S_p)$, one naturally asks that in which realization $g$ that $f_g(S_p,u)$ is equal to 1. For a realization $g$, let $\dis_g(u,v)$ be length of the shortest path from node $u$ to node $v$ in $g$, and, for a node set $V^{'}$, define that $\dis_g(V^{'},u)=\text{min}_{v \in V^{'}}\dis_g(v,u)$. A key lemma is shown as follows.
\begin{lemma}
\label{lemma: condition}
For a realization $g$ and $S_p$, a node $u$ will be activated by rumor in $g$ under $S_p$ (i.e. $f_g(S_p,u)=0$) if and only $\dis_g(S_r,u) \leq \dis_g(S_p,u)$ and $\dis_g(S_r,u) \neq +\infty$. 
\end{lemma}
\begin{proof}
See Appendix \ref{appendix: proof_lemma: condition}.
\end{proof}
As a corollary of Lemma \ref{lemma: condition}, our objective function is monotone submodular.

\begin{corollary}
\label{coro: submodular}
$f(S_p)$ is monotone submodular with respect to $S_p$.
\end{corollary}
\begin{proof}
According to Eq. (\ref{eq:f(S_p)}), it suffices to prove that $f_g(S_p,u)$ is monotone submodular for each $g$ and $u$. It is clear monotone as $\dis_g(S_p,u)$ does not increase when more nodes are added into $S_p$. Now we prove that $$f_g(V_1 \cup \{v\},u)-f_g(V_1,u) \geq f_g(V_2 \cup \{v\},u)-f_g(V_2,u)$$ holds for any $V_1 \subset V_2 \subseteq V$ and $v \in V \setminus V_2$. Since, $f_g(S_p,u)$ is either 0 or 1, we only need to prove that $$f_g(V_1 \cup \{v\},u)-f_g(V_1,u)=1$$ when $f_g(V_2 \cup \{v\},u)-f_g(V_2,u)$ is equal to 1. When $f_g(V_2 \cup \{v\},u)-f_g(V_2,u)=1$, $f_g(V_2 \cup \{v\},u)=1$ and $f_g(V_2,u)=0$, which means, by Lemma \ref{lemma: condition}, $\dis_g(V_2 \cup \{v\},u) < \dis_g(S_r,u)$, $\dis_g(V_2 ,u) \geq \dis_g(S_r,u)$ and $\dis_g(S_r,u) \neq +\infty$. Therefore,  $\dis_g(v,u) < \dis_g(S_r,u)$ and consequently $f_g(V_1 \cup \{v\},u)=1$. Furthermore, $f_g(V_1,u)=0$ because $V_1$ is a subset of $V_2$ and $f_g(V_2,u)=0$. As a result, $f_g(V_1 \cup \{v\},u)-f_g(V_1,u)$ is also equal to 1.
\end{proof}

According to Corollary \ref{coro: submodular}, it seems that we can use the greedy algorithm to maximize $f(S_p)$ according to Eq. (\ref{eq:f(S_p)}). Unfortunately, there are exponential number of realizations in $\mathcal{G}$ to consider and therefore the greedy algorithm does not run in polynomial. Alternatively, we utilize the reverse sampling technique to obtain an estimate of $f(S_p)$ and then maximize the estimate.

\subsection{A Sampling Method}

\begin{algorithm}[t]
\caption{ \textbf{Random R-tuple of $v$} $(G, S_r, v)$}\label{alg:r_tuple_v}
\begin{algorithmic}[1]
\State \textbf{Input:} $G$, $S_r$ and $v$; 
\State $V^* \leftarrow \emptyset$, $E_t \leftarrow \emptyset$ ,$E_f \leftarrow \emptyset$; 
\State $V_1 \leftarrow \{v\}$; 
\While {true}
	\If {$V_1=\emptyset$}
	\State $B=0$;  
	\State Return $(V^*, E_t, E_f, B)$;
	\EndIf 
	\If {$V_1 \cap S_r \neq \emptyset$}
		\State $B=1$;
		\State Return $(V^*, E_t, E_f, B)$;
	\EndIf 
	\State $V^* \leftarrow V^* \cup V_1$; 
	\State $V_2 \leftarrow V \setminus V^*$;
	\State $V_1 \leftarrow \emptyset$.
	\For {each edge $(u_1,u_2) \in E$ where $u_1 \in V^*$ and $u_2 \in V_2$} 
	\State $rand \leftarrow$ a random number from 0 to 1 generated in uniform; 
	\If {$rand  \leq p_{(u_2,u_1)}$}
	\State $V_1 \leftarrow V_1 \cup \{u_2\}$;
	\State $E_t \leftarrow E_t \cup \{(u_2,u_1)\}$;
	\Else
	\State $E_f \leftarrow E_f \cup \{(u_2,u_1)\}$;
	\EndIf 
	\EndFor
\EndWhile
\end{algorithmic}
%\vspace{-1mm}
\end{algorithm}

%\begin{definition}{(\textbf{Reverse network})} Given a network $G$, the reverse network $\overline{G}$ of $G$ is constructed as follows. $\overline{G}$ is identical to $G$ except that $p_{(u,v)}^{\overline{G}}=p_{(v,u)}^G$. In brief, $\overline{G}$ is obtained from $G$ by reversing each edge in $G$.
%\end{definition}
Our sampling method is designed based on the following objects.

\begin{definition}{(\textbf{Random R-tuple of $v$})} 
\label{def:reverse_set}
Generated by Algorithm \ref{alg:r_tuple_v}, a random reverse tuple $\mathcal{T}_v$ of node $v$ is a four tuple $(V^*, E_t, E_f, B)$ where $V^*$ is a node-set, $E_t$ and $E_f$ are edge-sets, and $B$ is a boolean variable. As shown in Alg. \ref{alg:r_tuple_v}, we start from $v$ and successively test whether the current \textit{in-neighbor} of the nodes in $V^*$ can be added to $V^*$ in a \textit{breadth first manner until one of the rumor seed is reached or no node can be furthered reached}. $E_t$ and $E_f$ are generated in line 18 and line 20, respectively. $V^*$ is set of nodes that are reachable to $v$. $B$ is set as $1$ if and only if some rumor seeds are encountered. We denote by $\mathcal{T}_v(V^*)$, $\mathcal{T}_v(E_t)$, $\mathcal{T}_v(E_f)$ and $\mathcal{T}_v(B)$ the four attributes of $\mathcal{T}_v$. 
\end{definition}

\begin{algorithm}[t]
\caption{ \textbf{Random R-tuple} $(G, S_r)$}\label{alg:r_tuple}
\begin{algorithmic}[1]
\State \textbf{Input:} $G$ and $S_r$;
\State Randomly select a node $v$ from $V$ in uniform; 
\State \small $(V^*,E_t,E_f,B) \leftarrow $ Algorithm \ref{alg:r_tuple_v}$(G, S_r, v)$;
\State Return $\T=(V^*,E_t,E_f,B)$
\end{algorithmic}
%\vspace{-1mm}
\end{algorithm}

\begin{definition}{(\textbf{Random R-tuple})} 
\label{def: x(S,T)}
Generated by Algorithm \ref{alg:r_tuple}, a random R-tuple $\T=(\T(V^*), \T(E_
t),\T(E_f), \T(B))$ is a random R-tuple $\T_v$ of $v$ generated by Alg. \ref{alg:r_tuple_v} where $v$ is selected from $V$ uniformly in random. For a node set $S \subseteq V$, let $x(S,\T)$ be a random variable over 0 and 1, where

$$x(S,\T) =
  \begin{cases}
  1 &  \hspace{0mm} \hspace{-0.5mm} \text{if $S \cap \T(V^*) \neq \emptyset $ or $\T(B)=0$} \\
  0 & \hspace{0mm} \hspace{-0.5mm} \text{else } 
  \end{cases}$$. 
\end{definition}

The following lemma is critical for the rest of the analysis in this section.
\begin{lemma}
\label{lemma:key}
$E[x(S,\T)]=f(S)/n$ for any $S \subseteq V$.
\end{lemma}
\begin{proof}
See Appendix \ref{appendix: proof_lemma:key}.
\end{proof}

\begin{algorithm}[t]
\caption{ \textbf{Node-Selection} $(V,R_l,k)$}\label{alg:set_cover}
\begin{algorithmic}[1]
\State \textbf{Input:} $V$, $R_l=\{\T_1,...,\T_l\}$ and $k$ 
\State Set $\T_i(V^*)$ as $V$ if $\T_i(B)=0$; 
\State $S^{'} \leftarrow\emptyset$;
\For {j = 1 : k} 
	\State Let $v$ be the node that covers the most number of sets in $\T_i(V^*)$;
	\State $S^{'} \leftarrow S^{'} \cup \{v\}$;
	\State Remove $v$ from each $\T_i(V^*)$; 
\EndFor
\State Return $S^{'}$;
\end{algorithmic}
%\vspace{-1mm}
\end{algorithm}

Suppose there is a set $R_l=\{\T_1,...,\T_l\}$ of $l$ random R-tuples each of which is obtained by Alg. \ref{alg:r_tuple} . For a set $S \subseteq V$ and $R_l$, let  $F(S,R_l)=\sum_{i=1}^{l} x(S,\T_i)$. Now let us consider the following problem 
\begin{problem}
\label{problem: set_cover}
Finding a node set $S$ with at most $k$ nodes such that $F(S,R_l)$ is maximized.
\end{problem}
Because $x(S,\T_i)$ is always 1 when $\T_i(B)=0$, we can take $\T(V^*)$ as the ground set $V$ such that $S \cap \T_i(V^*)$ is equal to 1 for any non-empty set $S$. Now Problem \ref{problem: set_cover} becomes the classic set cover problem and therefore the greedy algorithm shown in Algorithm \ref{alg:set_cover} produces a $(1-1/e)$-approximation \cite{nemhauser1978analysis}. For a given $R_l$, let $S^{'}$ be the set produced by Alg. \ref{alg:set_cover}. Then
\begin{equation}
\label{eq:greedy}
F(S^{'},R_l) \geq (1-1/e)\cdot F(S,R_l),
\end{equation}
for any $S \subseteq V$.

\subsection{Chernoff Bound}
In this paper, we use the Chernoff bound to analyze the error of estimating. Let $X_i$ be $l$ i.i.d random variables where $E(X_i)=\mu$. The Chernoff bound \cite{motwani2010randomized} states that 
\begin{equation}
\label{eq:chernoff_1}
\mathrm{Pr}\Big[\sum X_i - l \cdot \mu \geq \delta \cdot l \cdot  \mu  \Big]\leq \exp(-\frac{l \cdot \mu \cdot \delta^2}{2+\delta}),
\end{equation}
and
\begin{equation}
\label{eq:chernoff_2}
\mathrm{Pr}\Big[\sum X_i - l \cdot \mu \leq -\delta \cdot l \cdot \mu \Big]\leq \exp(-\frac{l \cdot \mu \cdot \delta^2}{2}),
\end{equation}  
for $0 < \delta < 1$.

\section{The algorithm}
In this section, we first discuss how to estimate the optimal value $OPT_k$ and then present the algorithm together with its analysis.
\label{sec:algorithm}
\subsection{Estimating $OPT_k$}
Estimating the optimal value of $f(S)$ is an important part of our algorithm. Suppose we have a set $R_l$ of random R-tuples. Intuitively, $\frac{n \cdot F(S^{'},R_l)}{l}$ should be a good choice because, by Eq. (\ref{eq:greedy}), it is close to $\frac{n \cdot F(S_k,R_l)}{l}$ with a guaranteed factor, and, according to Lemma \ref{lemma:key}, $\frac{n \cdot F(S_k,R_l)}{l}$ is an unbiased estimate of $OPT_k$. Because $OPT_k \in [1, n]$\footnote{$OPT_k$ is always no less than 1 because $k \geq 1$ and $S_p \cap S_r \neq \emptyset$.}, we design a statistic test which compares $\frac{n \cdot F(S^{'},R_l)}{l}$ with $n/2^i$ and terminates when they are sufficiently close to each other. The estimation process is shown in Algorithm \ref{alg:OPT_k} with tunable parameters $\delta>0$ and $N >0$. Let $OPT_k^*$ be the estimation produced by Alg. \ref{alg:OPT_k}.
\begin{algorithm}[t]
\caption{$OPT_k$-\textbf{Estimation} $(k,\delta, N)$}\label{alg:OPT_k}
\begin{algorithmic}[1]
\State \textbf{Input:} $(k,\delta, N)$
\State $R \leftarrow\emptyset$;
\State $\lambda_3 = \frac{n \cdot (2+\delta) \cdot \ln(N \cdot \binom{n}{k} \cdot \log n)}{\delta^2}$;
\For {i = 1 : $\log(n-1)$} 
	\State $x_i \leftarrow \frac{n}{2^i}$,  $l_i \leftarrow \frac{\lambda_3}{x_i}$;
	\While {$|R|\leq l_i$}
	\State Generate a random R-tuple and inset it into $R$;
	\EndWhile
	\State $S^{'} \leftarrow \text{Node-Selection}(V,R,k)$;
	\If {$\frac{n \cdot  F(S^{'},R)}{l_i} \geq (1+\delta) \cdot x_i$}; 
	\State $OPT_k^* =\frac{n \cdot F(S^{'},R)}{l_i \cdot (1+\delta)}$; 
	\State Return $OPT_k^*$;
	\EndIf
\EndFor
\end{algorithmic}
%\vspace{-1mm}
\end{algorithm}

First, we need to guarantee that $OPT_k^*$ is smaller than $OPT_k$. The following result shows that the terminate condition (i.e., line 9) leads that $OPT_k^*$ is smaller than $OPT_k$ with a high probability.

\begin{lemma}
\label{lemma: opt_k_1}
With probability at least $1-2/N$, Algorithm \ref{alg:OPT_k} produces an $OPT_k^*$ that is less $OPT_k$. 
\end{lemma}
\begin{proof}
See Appendix \ref{appendix: proof_lemma: opt_k_1}
\end{proof}

Second, it can be shown that $OPT_k^*$ is not too much less than $OPT_k$.
% which guarantees that $l^{*}=O(\dfrac{\lambda^*}{OPT_k})$.

\begin{lemma}
\label{lemma: opt_k_2}
With a probability at least $1-1/N$, Algorithm \ref{alg:OPT_k} produces an $OPT_k^*$ such that $OPT_k^* \geq \frac{(1-1/e) \cdot OPT_k}{2 \cdot (1+\delta)^2}$. 
\end{lemma}
\begin{proof}
See Appendix \ref{appendix: proof_lemma: opt_k_2}
\end{proof}

The above results are summarized as follows.

\begin{theorem}
\label{theorem:main}
With a probability at least $1-3/N$, Algorithm \ref{alg:OPT_k} returns an $OPT_k^*$, such that 
\begin{equation}
\label{eq: opt_k^*}
OPT_k \geq OPT_k^* \geq \frac{(1-1/e) \cdot OPT_k}{2(1+\delta)^2}
\end{equation}
\end{theorem}

\subsection{The Algorithm}

\begin{algorithm}[t]
\caption{ \textbf{R-tuple Based Randomized Rumor Blocking} \label{alg:algortihm}}
\begin{algorithmic}[1]
\State \textbf{Input:} $G$, $S_r$, $k$, $N$, $\delta_1$,  $\delta_2$ and $\delta_3$;
\State $OPT_k^* \leftarrow$ Algorithm \ref{alg:OPT_k} $(k,\delta_3,N)$.
\State $l_1=\frac{2n \ln N_1}{\delta_1^2 \cdot OPT_k^*}$;
\State $l_2=\frac{(2+\delta_2-(1-1/e)\delta_1)n\ln(N_2  \binom {n} {k})}{(\delta_2-(1-1/e)\delta_1)^2 \cdot OPT_k^*}$;
\State $l^*= \max(l_1,l_2)$;
\State Generate $l^*$ random R-tuples $R_{l^*}=\{T^1,...,T^{l^*}\}$;
\State Run Algorithm \ref{alg:set_cover} with input $R_{l^*}$ and $k$ to obtain a node set $S^{'}$;
\State Return $S^{'}$;
\end{algorithmic}
%\vspace{-1mm}
\end{algorithm}

Now we are ready to show the algorithm of rumor blocking. Let $R_l$ be a set of random R-tuples. According to Lemma \ref{lemma:key}, $n\cdot s(S,\T)$ is an unbiased estimate of $f(S)$ and therefore the $S$ that can maximize $\frac{n\cdot F(S,R_l)}{l}$ should be able to maximize $f(S)$ as long as $l$ is sufficiently large. The whole algorithm is given in Alg. \ref{alg:algortihm}. Let $N>0$, $\delta_1>0$, $\delta_2> (1-1/e)\cdot\delta_1$ and $\delta_3>0$ be some adjustable parameters. We first obtain an estimate $OPT_k^*$ of $OPT_k$ by Alg. \ref{alg:OPT_k} with input $(k, \delta_3, N)$ and set that

\begin{equation}
\label{eq:l_1}
l_1=\frac{2 \cdot n \cdot  \ln N_1}{\delta_1^2 \cdot OPT_k^*},
\end{equation}

\begin{equation}
\label{eq:l_2}
l_2=\frac{(2+\delta_2-(1-1/e) \cdot \delta_1) \cdot n \cdot \ln(N_2 \cdot   \binom {n} {k})}{(\delta_2-(1-1/e) \cdot \delta_1)^2 \cdot OPT_k^*},
\end{equation}
and
\begin{equation}
l^*=\max(l_1,l_2).
\end{equation}
Next, we generate $l^*$ random R-tuples $R_{l^*}=\{\T_1,...,\T_{l^*}\}$ by Alg. \ref{alg:r_tuple}. Finally, Alg. \ref{alg:algortihm} returns the set obtained by running Alg. \ref{alg:set_cover} with input $(V,R_{l^*},k)$. Let $S^{*}$ be the node set produced by Alg. \ref{alg:algortihm}. As mentioned early, $S^{*}$ should be a $(1-1/e)$-approximation to Problem \ref{problem:1} if the estimate is sufficiently accurate. In particular, we require the following accuracy of $F(S_k,R_{l^*})$ and $F(S^{*},R_{l^*})$. 
\begin{equation}
\label{eq: accuracy_S_k}
\frac{n}{l^*}\cdot F(S_k,R_{l^*}) \geq (1-\delta_1) \cdot OPT_k
\end{equation}
and
\begin{equation}
\label{eq: accuracy_S^'}
F(S^{*},R_{l^*})-\frac{l^*}{n} \cdot f(S^{*}) \leq (\delta_2-(1-1/e) \cdot \delta_1) \cdot \frac{l^*}{n} \cdot OPT_k
\end{equation}
The following lemma shows $S^{*}$ is a  $(1-1/e-\delta_2)$-approximation if Eqs. (\ref{eq: accuracy_S_k}) and (\ref{eq: accuracy_S^'}) hold simultaneously.

\begin{lemma}
\label{lemma: ratio}
With Eqs. (\ref{eq: accuracy_S_k}) and (\ref{eq: accuracy_S^'}), $f(S^{*})$ is no less than $(1-1/e-\delta_2) \cdot OPT_k$ 
\end{lemma}
\begin{proof}
Let $\delta_*=(\delta_2-(1-1/e) \cdot \delta_1)$. By Eq. (\ref{eq: accuracy_S^'}),
\begin{eqnarray*}
f(S^{*})&\geq &\frac{n}{l^{*}} \cdot F(S^{*},T)-\delta_* \cdot OPT_k\\
&&\{\text{By Eq. (\ref{eq:greedy})}\}\\
&\geq& \frac{n}{l^{*}} \cdot (1-1/e) \cdot F(S_k,T)-\delta_* \cdot OPT_k \\
&&\{\text{By Eq. (\ref{eq: accuracy_S_k})}\}\\
&\geq& (1-\delta_1) \cdot (1-1/e) \cdot OPT_k-\delta_* \cdot  OPT_k\\
&=& (1-1/e-\delta_2) \cdot OPT_k
\end{eqnarray*}
\end{proof}

Setting $l^*$ as $\max(l_1,l_2)$ is able to guarantee that Eqs. (\ref{eq: accuracy_S_k}) and (\ref{eq: accuracy_S^'}) hold with a provable probability provided that Eq. (\ref{eq: opt_k^*}) holds, which is shown in the following two lemmas.

\begin{lemma}
\label{lemma:1}
Eq. (\ref{eq: accuracy_S_k}) holds with probability at least $1-1/N$ if $l \geq l_1$ and Eq. (\ref{eq: opt_k^*}) holds.
\end{lemma}

\begin{proof}
See Appendix \ref{appendix: proof_lemma:1}.
\end{proof}

\begin{lemma}
\label{lemma:2}
Eq. (\ref{eq: accuracy_S^'}) holds with probability at least $1-1/N$ if $l \geq l_2$ and Eq. (\ref{eq: opt_k^*}) holds.

\end{lemma}
\begin{proof}
See Appendix \ref{appendix: proof_lemma:2}.
\end{proof}

%\begin{algorithm}[t]
%\caption{ \textbf{Positive-Seed-Detection} $(G,\lambda^{*},OPT_k,S_r,k)$}\label{alg:detection}
%\begin{algorithmic}[1]
%\State \textbf{Input:} $G$, $\lambda^{*}$, $OPT_k$, $S_r$ and $k$
%\State Generate $l^*=\frac{\lambda^{*}}{OPT_k}$ random reverse sets %$R_{l^*}=\{T^1,...,T^{l^*}\}$
%\State Run Algorithm \ref{alg:set_cover} with input $R_{l^*}$ and $k$ to obtain a node set %$S^{'}$;
%\State Return $S^{'}$;
%\end{algorithmic}
%\setlength{\textfloatsep}{0pt}
%\end{algorithm}
The above analysis is summarized as follows.
\begin{theorem}
\label{theorem:approximation}
With probability at least $1-5/N$, $f(S^{*}) \geq (1-1/e-\delta_2) \cdot OPT_k$.
\end{theorem}
\begin{proof}
According to Theorem \ref{theorem:main}, Lemmas \ref{lemma:1} and \ref{lemma:2}, by the union bound, Eqs. (\ref{eq: accuracy_S_k}) and (\ref{eq: accuracy_S^'}) holds with probability at least $1-5/N$, and this theorem follows immediately from Lemma \ref{lemma: ratio}.
\end{proof}

\subsection{Running Time}
Now let us consider the running time of Alg. \ref{alg:algortihm}. Let $TIME$ be the expected running time of Alg. \ref{alg:r_tuple}. The following lemma shows $TIME$ can be bounded by the objective value of the optimal solution.
\begin{lemma}
\label{lemma:TIME}
$TIME\leq\frac{m}{n} \cdot OPT_1$.
\end{lemma}
\begin{proof}
See Appendix \ref{appendix: proof_lemma:TIME}.
\end{proof}

\begin{theorem}
\label{theorem:time}
Alg. \ref{alg:algortihm} runs in $O(\frac{km\ln n}{\delta_2^2})$.
\end{theorem}
\begin{proof}
Alg. \ref{alg:set_cover} can be implemented to run in time linear to the total size of its input \cite{vazirani2013approximation}. Alg. \ref{alg:OPT_k} invokes Alg. \ref{alg:set_cover} with input size from $O(\ln \ln n)$ to $\lambda/x_i$ where $i$ is the index of the last iteration and the input size is doubled in each iteration. By Lemma \ref{lemma: opt_k_1}, the total number R-tuples generated by Alg. \ref{alg:OPT_k} is $O(\frac{n\ln n}{OPT_k})$ and, by Lemma \ref{lemma:TIME}, line 2 of Alg. \ref{alg:algortihm} runs in $O(m\ln n)$. The running time of lines 6 and 7 is dominated by that of line 2. Therefore, the running time of Alg. \ref{alg:algortihm} is $O(\frac{km\ln n}{\delta_2^2})$ by taking $N$, $\delta_1$, $\delta_3$ as constants and assuming $m \geq n$.
\end{proof}

\subsection{Parameters}
As shown in Theorem. \ref{theorem:approximation}, $N$ and $\delta_2$ controls the success probability and the approximation ratio, respectively. When $\delta_2$ is fixed, we select the $\delta_1$ such that $l^*$ can be minimized to reduce the running time. As shown in Alg. \ref{alg:OPT_k}, $\delta_3$ decides the value of $OPT_k^*$. When $\delta_3$ is getting larger, Alg. \ref{alg:OPT_k} takes less time while Alg. \ref{alg:algortihm} takes more time because, by Theorem \ref{theorem:main},  $OPT_k^*$ becomes smaller and consequently $l^*$ becomes larger. In experiments, we simply set that $\delta_2=\delta_3=0.1$ and $N=n$.

%\textbf{The RBR algorithm}. According to the analysis above, our randomized rumor blocking framework is shown in Algorithm \ref{alg:algortihm}, denoted as reverse-tuple based randomized (RBR) algorithm. We first apply Algorithm \ref{alg:OPT_k} to obtain an estimate $OPT_k^*$ of $OPT_k$ and then find a seed set $S^{'}$ using Algorithm \ref{alg:detection} with the input $(G,\lambda^{*},OPT_k^*,S_r,k)$. By Theorems \ref{theorem:approximation} and \ref{theorem:main} and Corollary \ref{corollary:time}, $f(S^{'}) \geq (1-1/e-\delta_2) \cdot OPT_k$ holds with a high probability and the expected running time of Algorithm \ref{alg:algortihm} is $O(\frac{km\ln n}{\delta_2^2})$. In the experiment, we simply set that $N_1=N_2=N=n$, $\delta_2=\delta_3=0.1$. Now the only undetermined parameter is $\delta_1$. According to Algorithm \ref{alg:detection}, we select the $\delta_1$ such that $ \lambda^*$ can be minimized. 

\begin{table}[t]
\renewcommand{\arraystretch}{1.2}
\centering
{\begin{tabular}{ |C{1.8cm} | C{1.0cm}| C{1.0cm} |C{2.8cm}|}
\hline
\textbf{Dataset} 	& \textbf{Node\#} 		& \textbf{Edge\#} 		& \textbf{Average out-Degree} \\   
\hline
Power2500 	& 2.5K		& 26K 		& 20.8 \\
\hline
Wiki 	& 7K		& 30K 		& 12.0 \\
\hline
Epinions 	& 75K		& 508K 		& 13.4 \\
\hline
Youtube 	& 1.1M		& 6.0M 		& 5.4 \\
\hline
\end{tabular}}
\caption{\textbf{Datasets}}
\label{table:datasets}
%\label{table:state1}
%\caption{States of edge $(v_{i,j},v_{0,i})$, for $i>0$.}
\vspace{-1mm}
\end{table}

\begin{figure*}[t]
\centering
\subfloat[Power2500 under CP model]{\label{fig:pl_01}\includegraphics[width=0.23\textwidth]{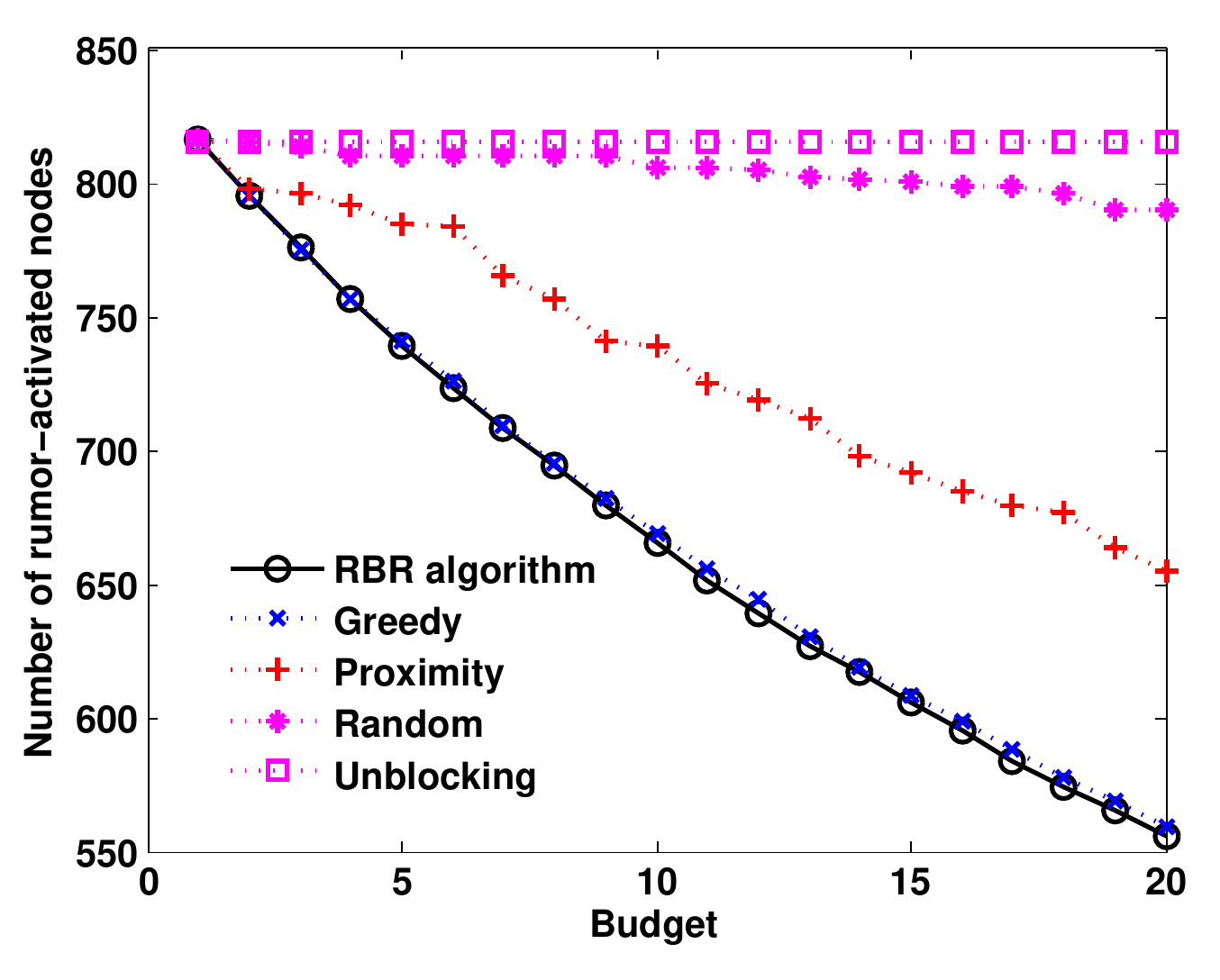}} \hspace{0mm} \ %vspace{-1mm}
\subfloat[Wiki under CP model]{\label{fig:wiki_01}\includegraphics[width=0.23\textwidth]{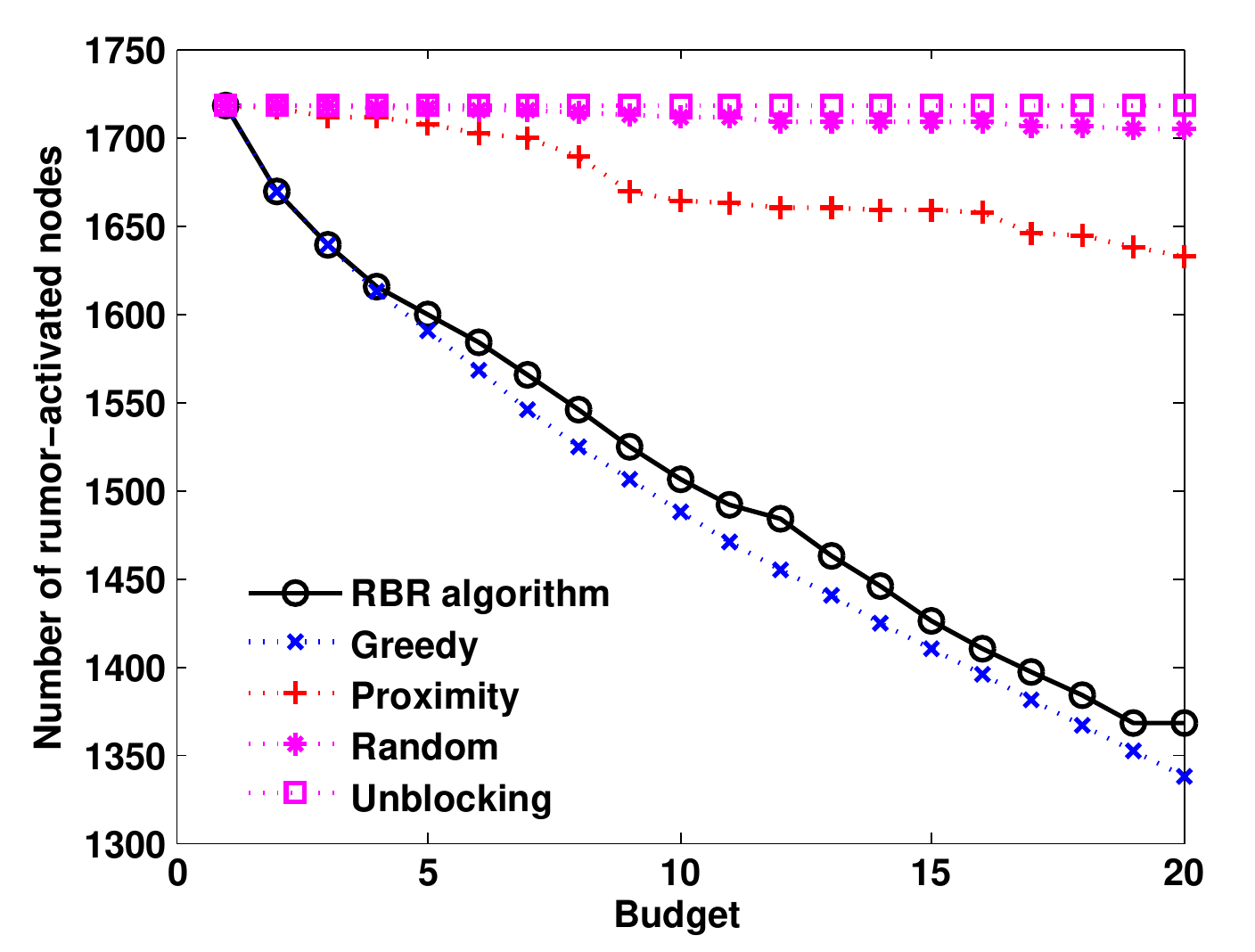}} \hspace{0mm} %\vspace{-1mm}
\subfloat[Epinions under CP model]{\label{fig:epin_01}\includegraphics[width=0.23\textwidth]{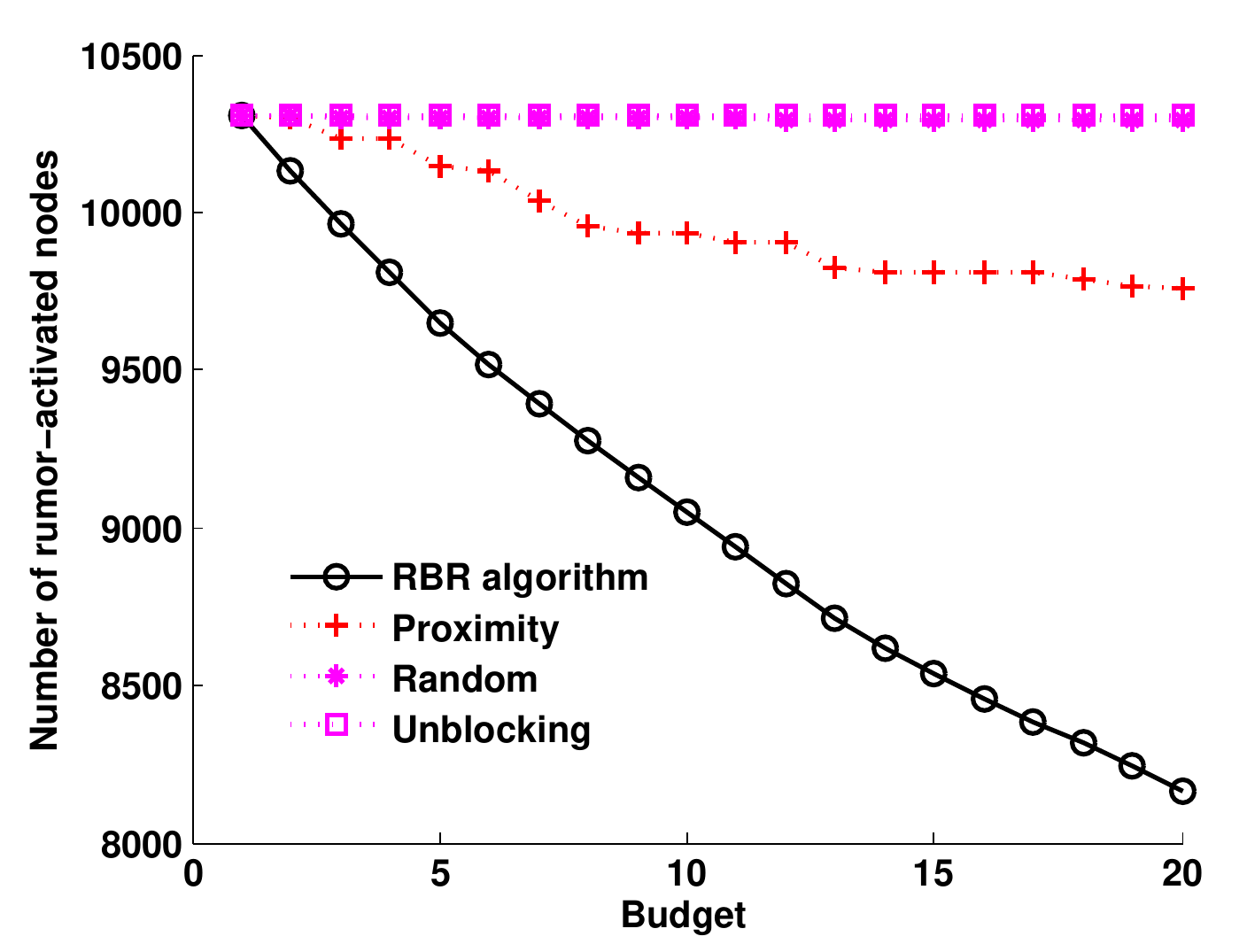}} \hspace{0mm} %\vspace{-1mm}
\subfloat[Youtube under CP model]{\label{fig:youtube_01}\includegraphics[width=0.23\textwidth]{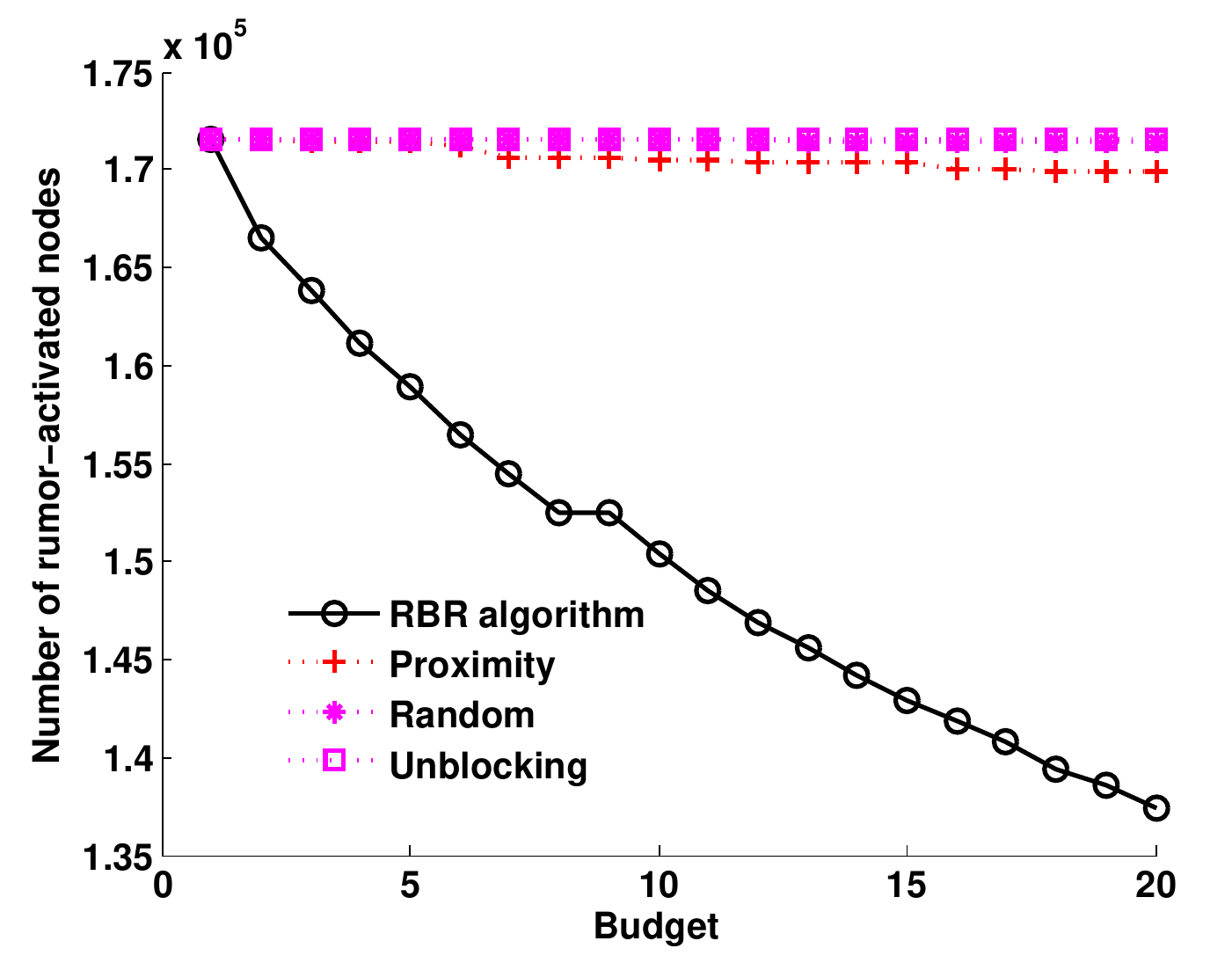}} \hspace{0mm} %\vspace{-1mm}
%\vspace{-3mm}

\subfloat[Power2500 under WC model]{\label{fig:pl_wc}\includegraphics[width=0.23\textwidth]{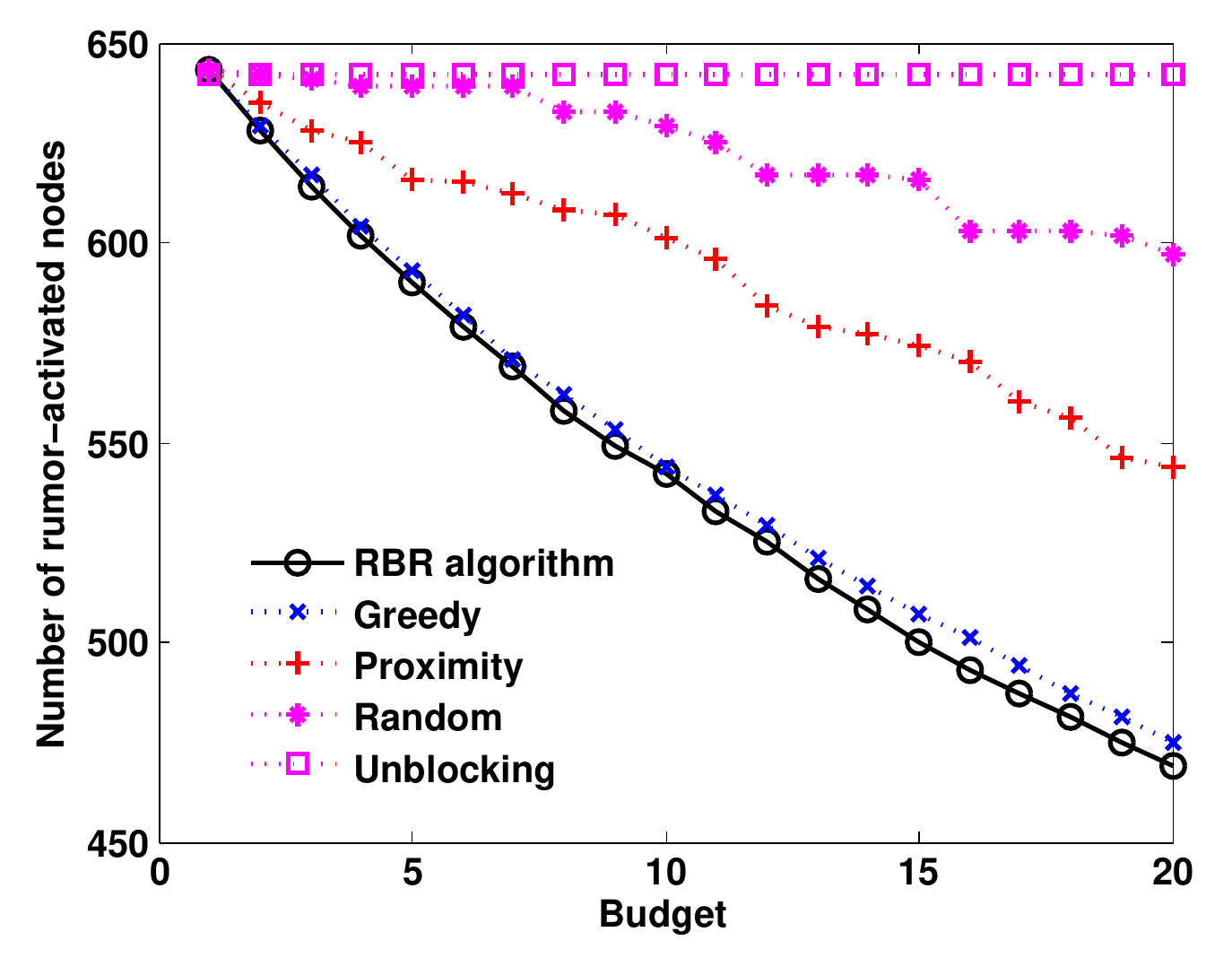}} \hspace{0mm} \ %vspace{-1mm}
\subfloat[Wiki under WC model]{\label{fig:wiki_wc}\includegraphics[width=0.23\textwidth]{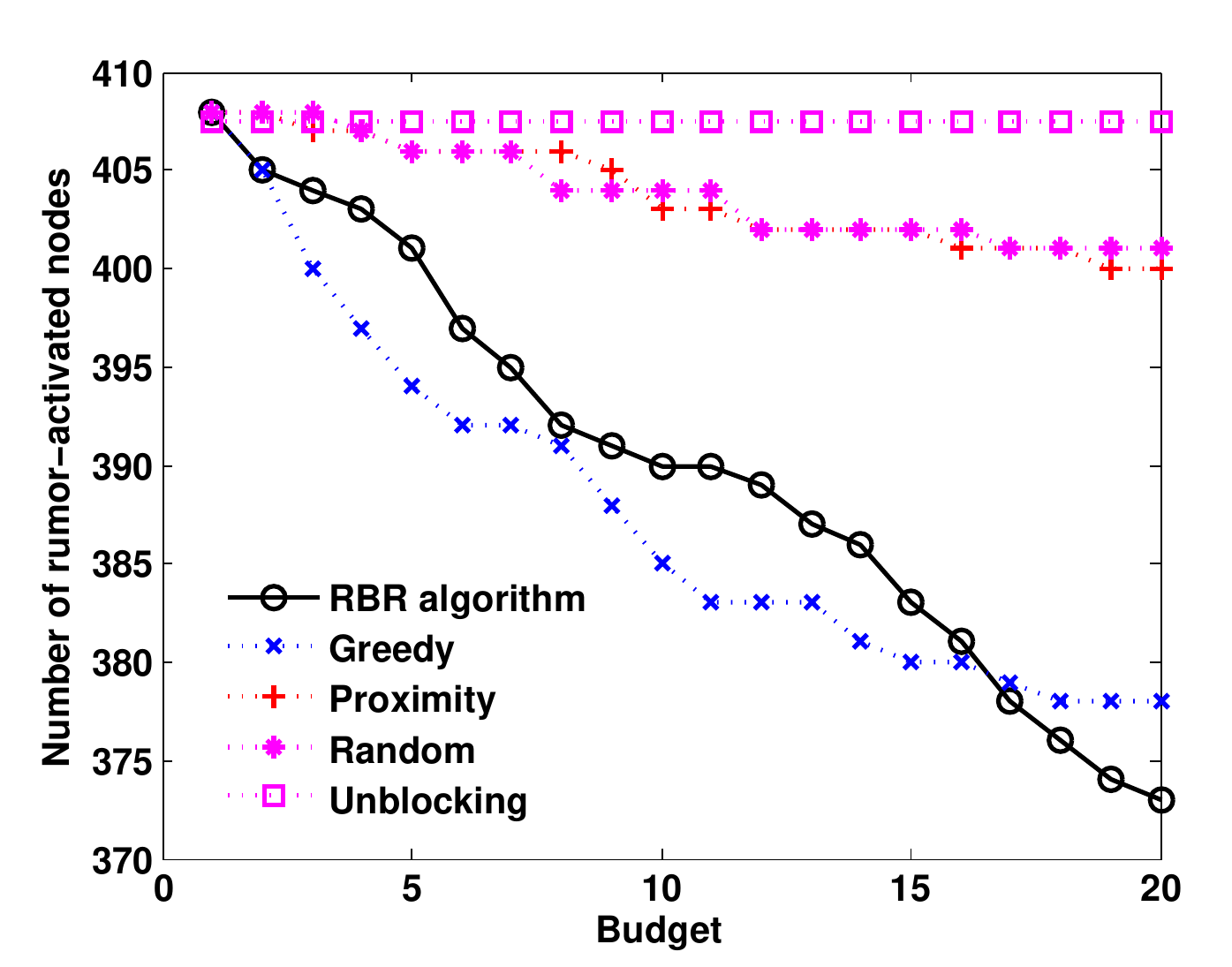}} \hspace{0mm} %\vspace{-1mm}
\subfloat[Epinions under WC model]{\label{fig:epin_wc}\includegraphics[width=0.23\textwidth]{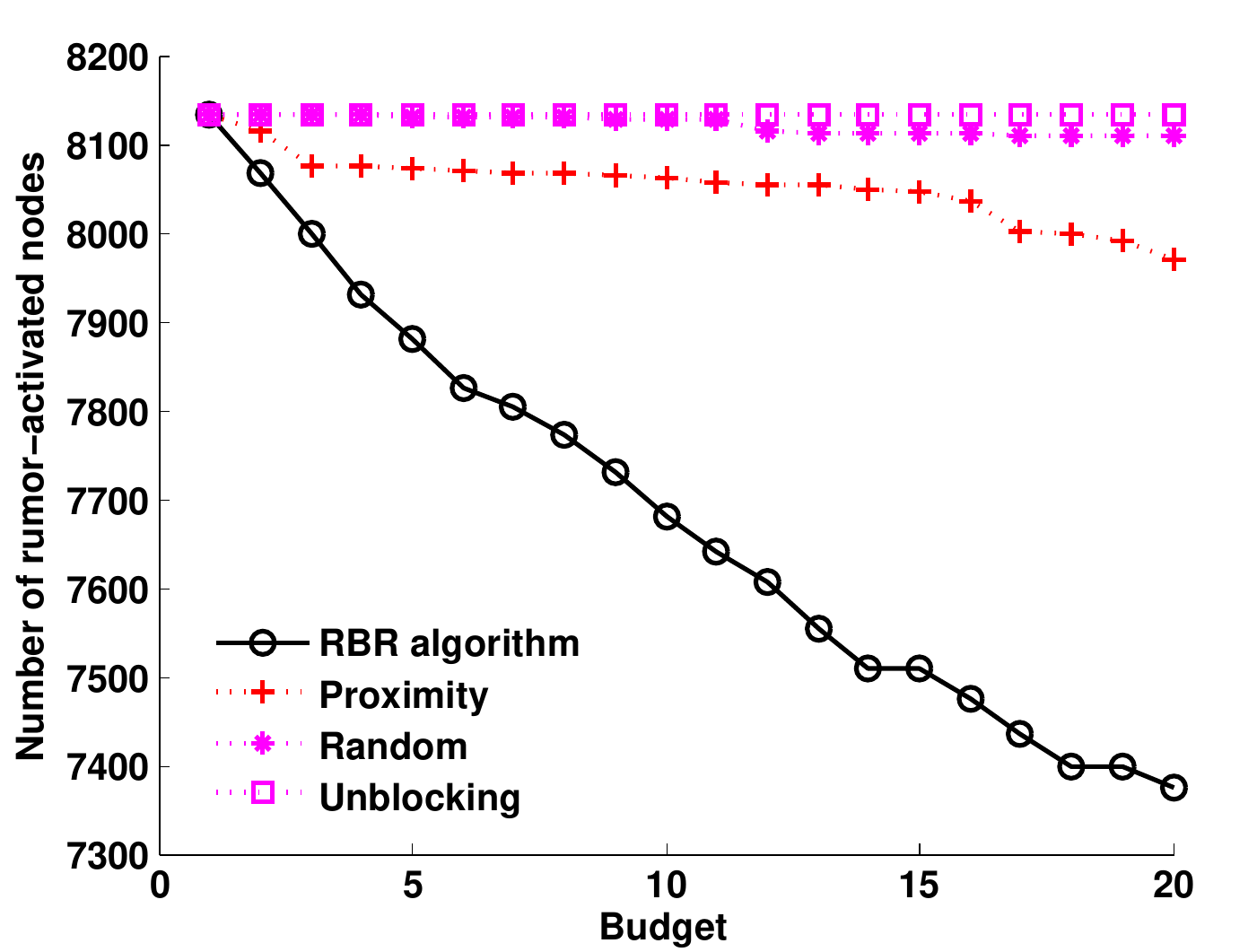}} \hspace{0mm} %\vspace{-1mm}
\subfloat[Youtube under WC model]{\label{fig:youtube_wc}\includegraphics[width=0.23\textwidth]{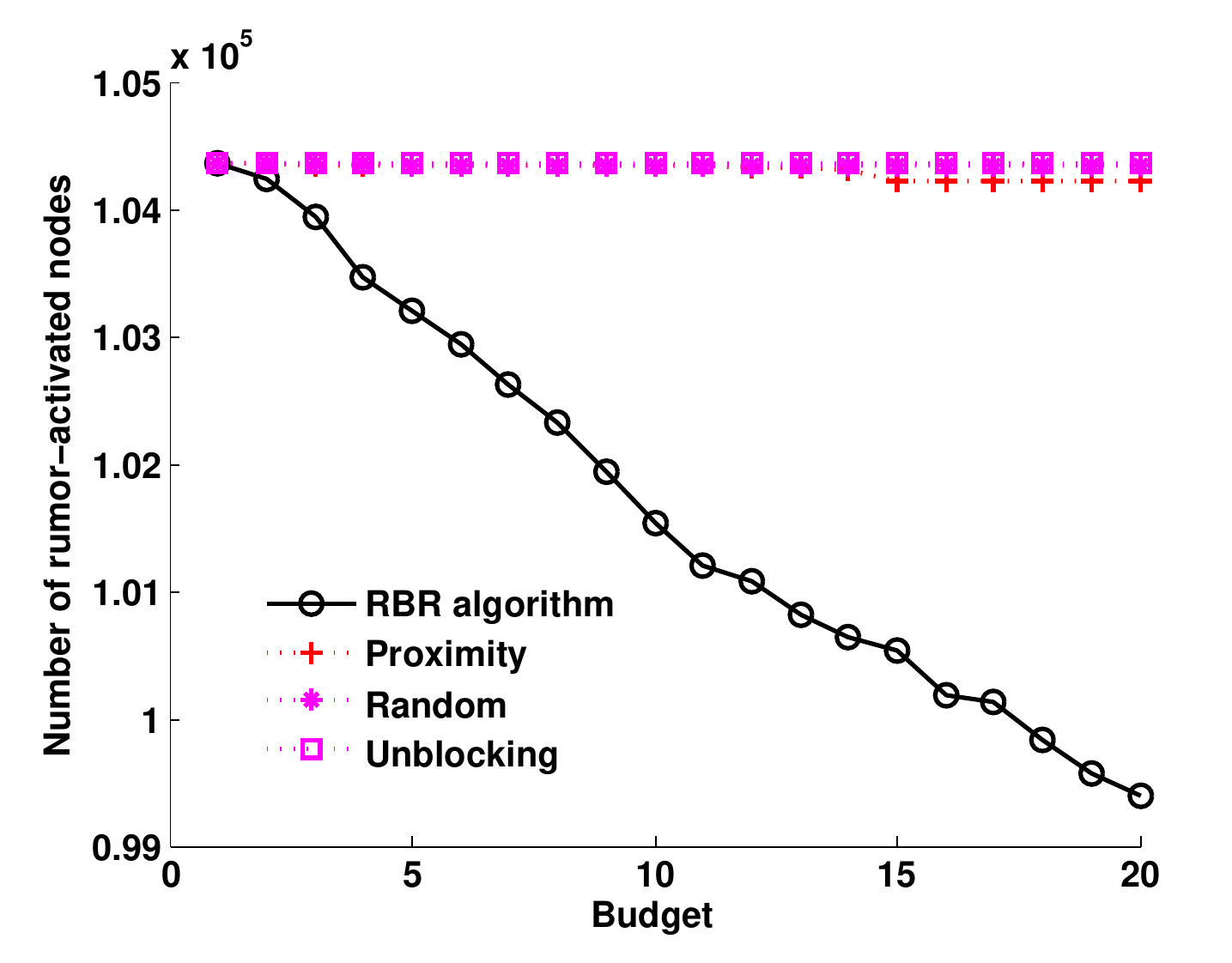}} \hspace{0mm} %\vspace{-1mm}
\caption{Experimental results} 
\label{fig: exp1}
%\vspace{-5mm}
\end{figure*}

\section{Experiment}
\label{sec:exp}
In this section, we evaluate the performance of RBR algorithm with respect to the state-of-the-art method and other heuristics. Besides, we also discuss the running time of the considered algorithms.

In our experiments, we employ four datasets, \textbf{Power2500}, \textbf{Wiki}, \textbf{Epinion} and \textbf{Youtube}, scaling from small to large. Power2500 is a synthetic power-law network generated by DIGG \cite{cowendigg}. It has been shown that the power-law distribution is one of the most important characteristics of social networks \cite{clauset2009power}. Wiki is a who-votes-on-whom network extracted from the vote history data of Wikipedia\footnote{https://www.wikipedia.org/}. Epinions is a who-trust-whom online social network extracted from the consumer review site Epinions.com. Youtube is a social network of a video-sharing website. Wiki, Epinions and Youtube are provided by the SNAP\cite{snapnets}. The basic statistics of the above datasets are shown in Table \ref{table:datasets}. The probability on the edges is either uniformly set as 0.1 or $p_{(u,v)}$ is set as $1/|N^-_v|$. These two settings are denoted as constant probability (CP) model and weighted cascade (WC) model. The above datasets together with the probability settings are widely used in the prior works.

%\begin{figure}[t]
%\centering
%\subfloat[Power2500]{\label{fig:lambda_pl}\includegraphics[width=0.23\textwidth]{images/exp/time_pl2500_01.pdf}} \hspace{0mm}  %vspace{-1mm}
%\subfloat[$p_{(u,v)}=0.1$]{\label{fig:time_pl2500_01_log}\includegraphics[width=0.23\textwidth]{images/exp/embed_time_pl2500_01_log.pdf}} \hspace{0mm} %\vspace{-1mm}
%\subfloat[Power2500]{\label{fig:lambda_pl}\includegraphics[width=0.23\textwidth]{images/exp/time_pl2500_degree.pdf}} \hspace{0mm}  %vspace{-1mm}
%\subfloat[$p_{(u,v)}=1/d_v$]{\label{fig:time_pl2500_degree_log}\includegraphics[width=0.23\textwidth]{images/exp/embed_time_pl2500_degree_log.pdf}} \hspace{0mm} \vspace{-1mm}
%\caption{\small Running time  on Power2500 under the log-normal scale.} 
%\vspace{-5mm}
%\label{fig:time}
%\end{figure}

\begin{table*}[t]
\centering
\renewcommand{\arraystretch}{1.2}
\begin{tabular}{|c|c c|c c|c c|c c|p{0.45cm} p{0.45cm}|p{0.45cm} p{0.45cm}|p{0.45cm} p{0.45cm}|p{0.45cm} p{0.45cm}|}
\hline
{\textbf{Dataset}}&\multicolumn{4}{c|}{\textbf{Power2500}} & \multicolumn{4}{c|}{\textbf{Wiki}} & \multicolumn{4}{c|}{\textbf{Epinions}} & \multicolumn{4}{c|}{\textbf{Youtube}}\\
\hline
{\textbf{Model}}& \multicolumn{2}{c|}{\textbf{CP}} & \multicolumn{2}{c|}{\textbf{WC}} & \multicolumn{2}{c|}{\textbf{CP}} & \multicolumn{2}{c|}{\textbf{WC}} & \multicolumn{2}{c|}{\textbf{CP}} & \multicolumn{2}{c|}{\textbf{WC}}& \multicolumn{2}{c|}{\textbf{CP}} & \multicolumn{2}{c|}{\textbf{WC}} \\
\hline
&time& \# R & time& \# R & time& \# R & time& \# R & time& \# R & time& \# R & time& \# R & time& \# R\\ 
\cline{1-17}
RBR&0.31s& 220K &  0.36s  & 211K & 0.33s & 236K & 0.20s & 208K & 1.7s & 248K & 0.68s &  274K& 42s & 348K & 9.4s & 336K\\ \hline
Greedy&99min& n/a &  48min  & n/a & 832min & n/a & 318min & n/a & n/a & n/a & n/a &  n/a& n/a & n/a & n/a & n/a\\ \hline
\end{tabular}
\caption{Running time and the number of R-tuples when $k=20$.}
\label{table: exp1}
\end{table*}

%\begin{table}
%  \centering
%  \renewcommand{\arraystretch}{1.2}
%  \begin{tabular}{|p{2cm}|c|c|c|c|}
%    \hline
%    \multirow{2}{5cm}{\textbf{Algorithm}} & \multicolumn{2}{c|}{\textbf{Power2500}} & \multicolumn{2}{c|}{\textbf{Wiki}}\\
    % \hline
    % \textbf{Inactive Modes} & \textbf{Description}\\
%    \cline{2-5}
%    & CP & WC & CP & WC\\
    %\hhline{~--}
%    \hline
%    Greedy & 99 min & 48 min & 832 min & 318 min \\ \hline
%  \end{tabular}
%  \caption{Running time of RBR and Greedy on Power2500 and Wiki.}
%\end{table}

We consider four rumor blocking algorithms shown as follows:
\begin{itemize}
\item \textbf{RBR algorithm.} This is the algorithm proposed in this paper. We set $\delta_2=\delta_3=0.1$ and $N=n$ by default. 
\item \textbf{Greedy.} This is the state-of-the-art rumor blocking algorithm using the Monte Carlo simulation. 2,000 simulations are used for each estimation. Greedy is only tested on small graphs, Power2500 and Wiki.
\item \textbf{Proximity.} This is a popular heuristic algorithm which selects the out-neighbors of the rumor seed nodes as the positive seed nodes. In particular, we give an index to each node and select the neighbors with the highest index.
\item \textbf{Random.} This is a baseline method where the positive seed nodes are randomly selected.
\item \textbf{Unblocking.} This is the base case when there is no positive cascade. 
\end{itemize}

In our experiments, the rumor seed nodes are selected from the nodes with the highest degree. The number of rumor seed nodes is set as 20 and the budget of positive seed set is selected from $1$ to $20$. The function value $f(S)$ of the seed set $S$ produced by each algorithm is finally evaluated by $n\cdot F(S,R_l)/l$ with $l=1,000,000$ where the R-tuples are separately generated.  

We conduct two series of experiments. In the first experiment, we evaluate the performance of RBR algorithm. In the second experiment, we investigate how many R-tuples that RBR algorithm needs to produce a high quality seed set.    
%As mentioned in Sec. \ref{sec:algorithm}, the parameter $\delta_1$ is selected to minimize $\lambda^*$ according to Eqs. (\ref{eq:l_1}) and (\ref{eq:l_2}). 
%For the considered three networks, the relationship between $\lambda_1$ and $\lambda_1$ is shown in Fig \ref{fig:lambda}.

%\begin{figure}[t]
%\centering
%\subfloat[Power2500]{\label{fig:lambda_pl}\includegraphics[width=0.15\textwidth]{images/exp/lambda_pl.pdf}} \hspace{0mm}  %vspace{-1mm}
%\subfloat[Wiki]{\label{fig:lambda_wiki}\includegraphics[width=0.15\textwidth]{images/exp/lambda_wiki.pdf}} \hspace{0mm} %\vspace{-1mm}
%\subfloat[Epinions]{\label{fig:lambda_epin}\includegraphics[width=0.15\textwidth]{images/exp/lambda_epin.pdf}} \hspace{0mm} %\vspace{-1mm}

%\caption{\small $\lambda^*$ vs $\lambda_1$ on three networks} 
%\label{fig:lambda}
%\end{figure}

\subsection{Results}
The analysis of the experimental results of the two series of experiments are shown in the following two subsections, respectively.
\subsubsection{Experiment \Romannum{1}}
In the first set of experiments, we compare the RBR algorithm with other existing methods. The experimental results are shown in Fig. \ref{fig: exp1} and Table \ref{table: exp1}.

The results on graph Power2500 are shown in Figs. \ref{fig:pl_01} and \ref{fig:pl_wc}. One can see that RBR algorithm and Greedy have the same performance with respect to $f()$. Nevertheless, RBR algorithm is more efficient than Greedy with respect to running time, as shown in Table \ref{table: exp1}. For example, under the CP model on Power2500 with $k=20$, RBR takes 0.31 second while Greedy takes about 1.3 hour. 

The results on the Wiki dataset are shown in Figs. \ref{fig:wiki_01} and \ref{fig:wiki_wc}. Under the CP model, RBR algorithm is able to reach at least 97.98\% blocking effect of the Greedy algorithm with respect to $f()$. Under the WC model, Greedy performs better than RBR does until $k$ is larger than 16. Recall that 2,000 simulations are used by Greedy for each estimation. Such a phenomenon suggests that, when $k$ is larger than 16, more simulations are required to maintain the accuracy of the estimates so that Greedy is able to achieve the (1-1/e)-approximation. However, as shown in Table \ref{table: exp1}, Greedy has already been very time consuming on Wiki with 2,000 simulations, and therefore using more simulations is not a good choice even though it may increase the quality of the produced seed set.  Despite that Wiki is larger than Power2500, comparing Figs. \ref{fig:pl_wc} and \ref{fig:wiki_wc}, one can see that when there is no positive cascade, 20 rumor seed nodes result 410 and 650 rumor-activated nodes on Wiki and Power2500, respectively, which indicates that the dense of the network has more impact on the influence diffusion than the network scale does.

The results on the Epinions dataset are shown in Figs. \ref{fig:epin_01} and \ref{fig:epin_wc}. One can see that on the large network RBR algorithm is superior to other heuristics by a significant margin. Under the CP model  when $k=20$, the RBR algorithm can protect about 2,000 users while Proximity protects 500 nodes. On Youtube, as shown in Figs. \ref{fig:youtube_01} and \ref{fig:youtube_wc}, RBR is still effective but other heuristics can hardly protect any node.  

\begin{figure}[pt]
\centering
\captionsetup{justification=centering}
\subfloat[Power2500 under CP model.]{\label{fig:pl_01_r}\includegraphics[width=0.23\textwidth]{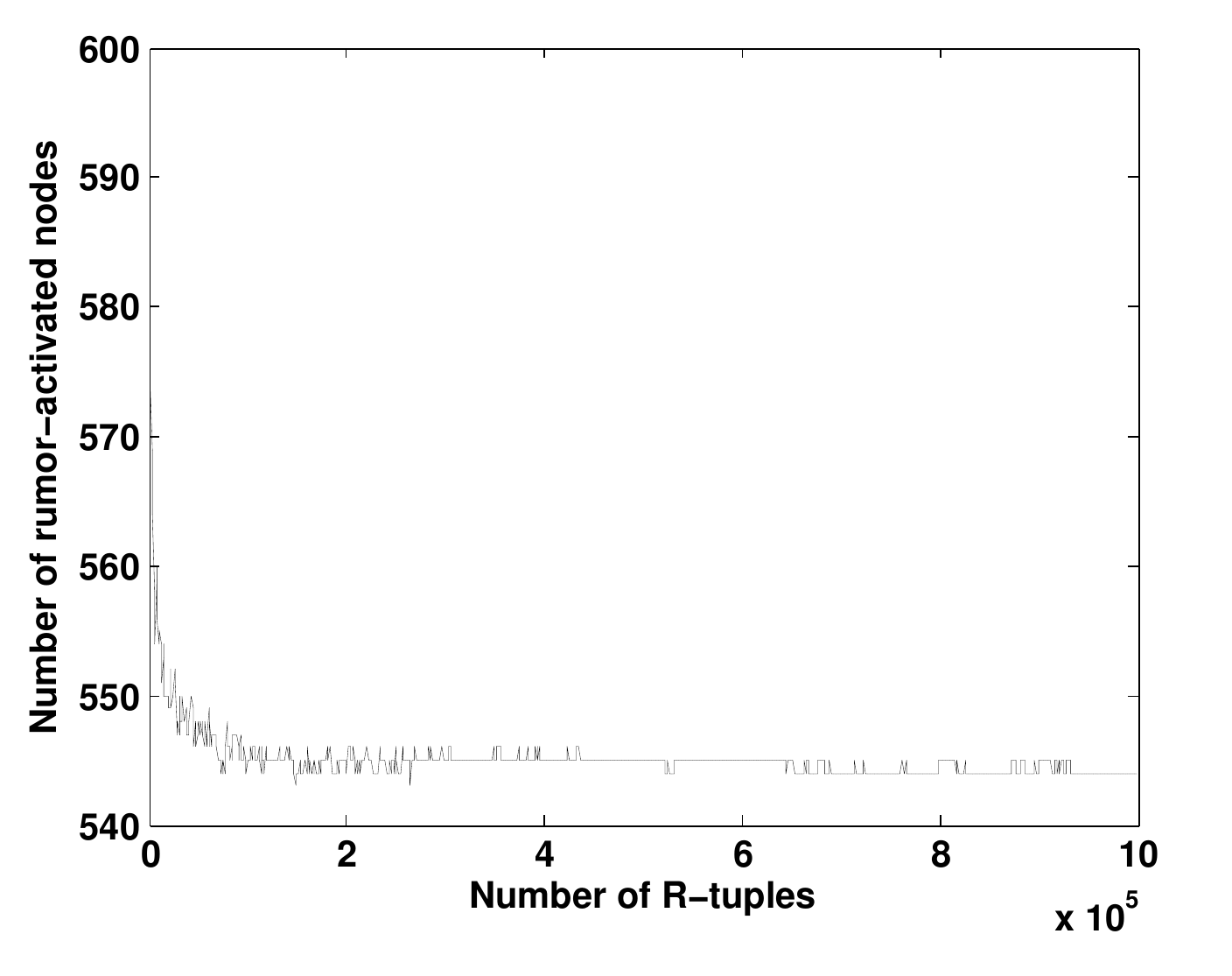}} \hspace{0mm}
\subfloat[Power2500 under WC model.]{\label{fig:pl_wc_r}\includegraphics[width=0.23\textwidth]{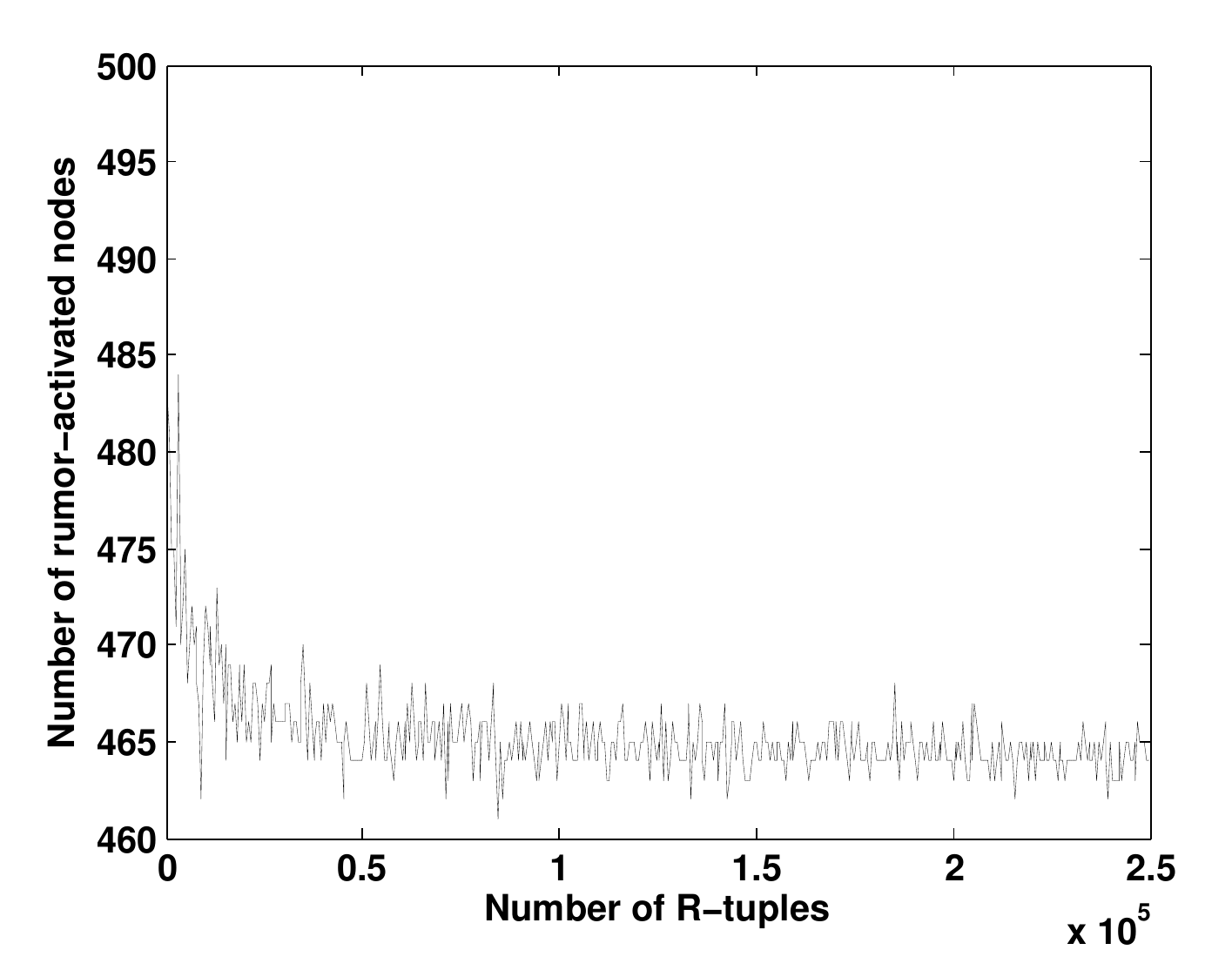}} \hspace{0mm}

\subfloat[Wiki under CP model.]{\label{fig:wiki_01_r}\includegraphics[width=0.23\textwidth]{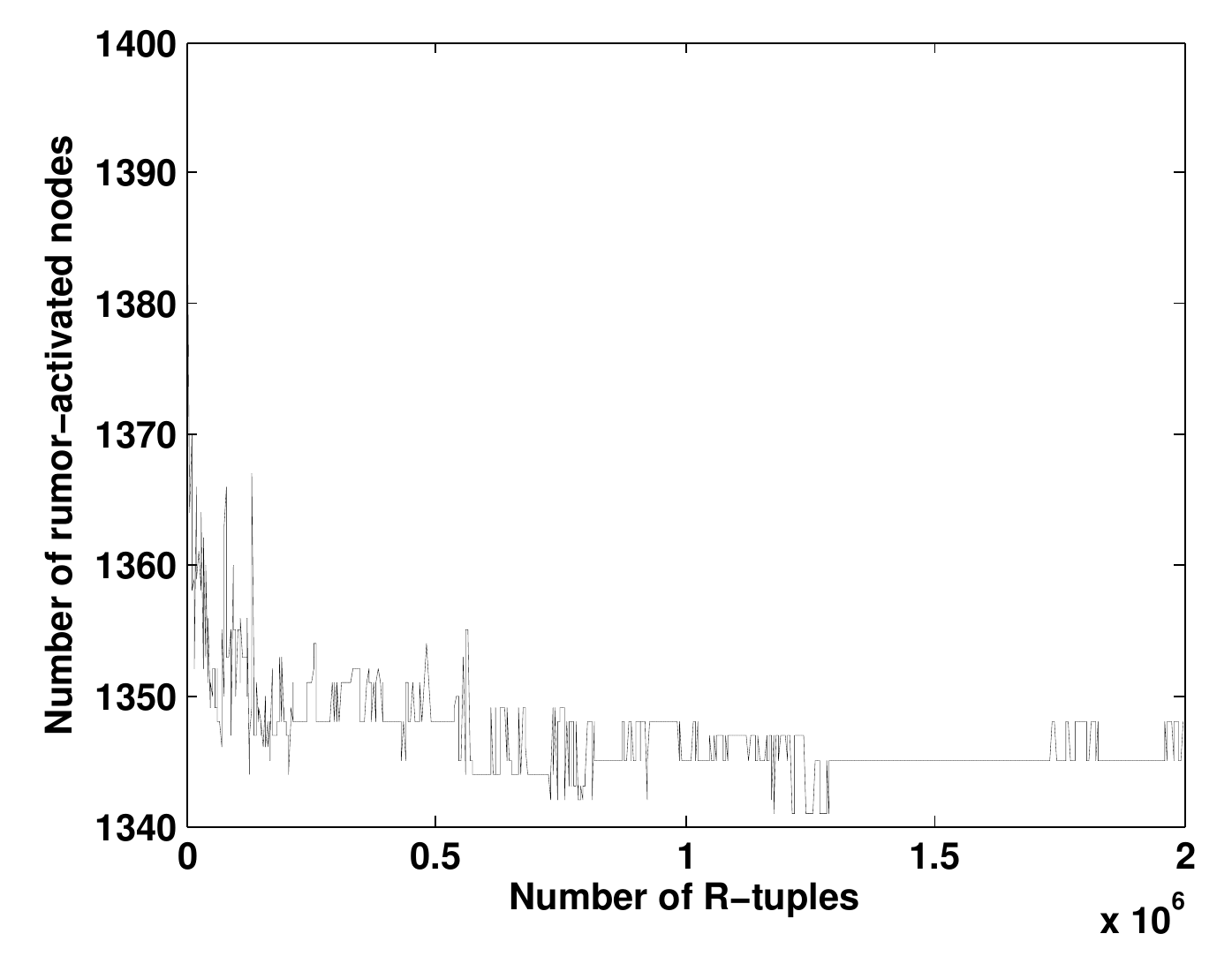}} \hspace{0mm}
\subfloat[Wiki under WC model.]{\label{fig:wiki_wc_r}\includegraphics[width=0.23\textwidth]{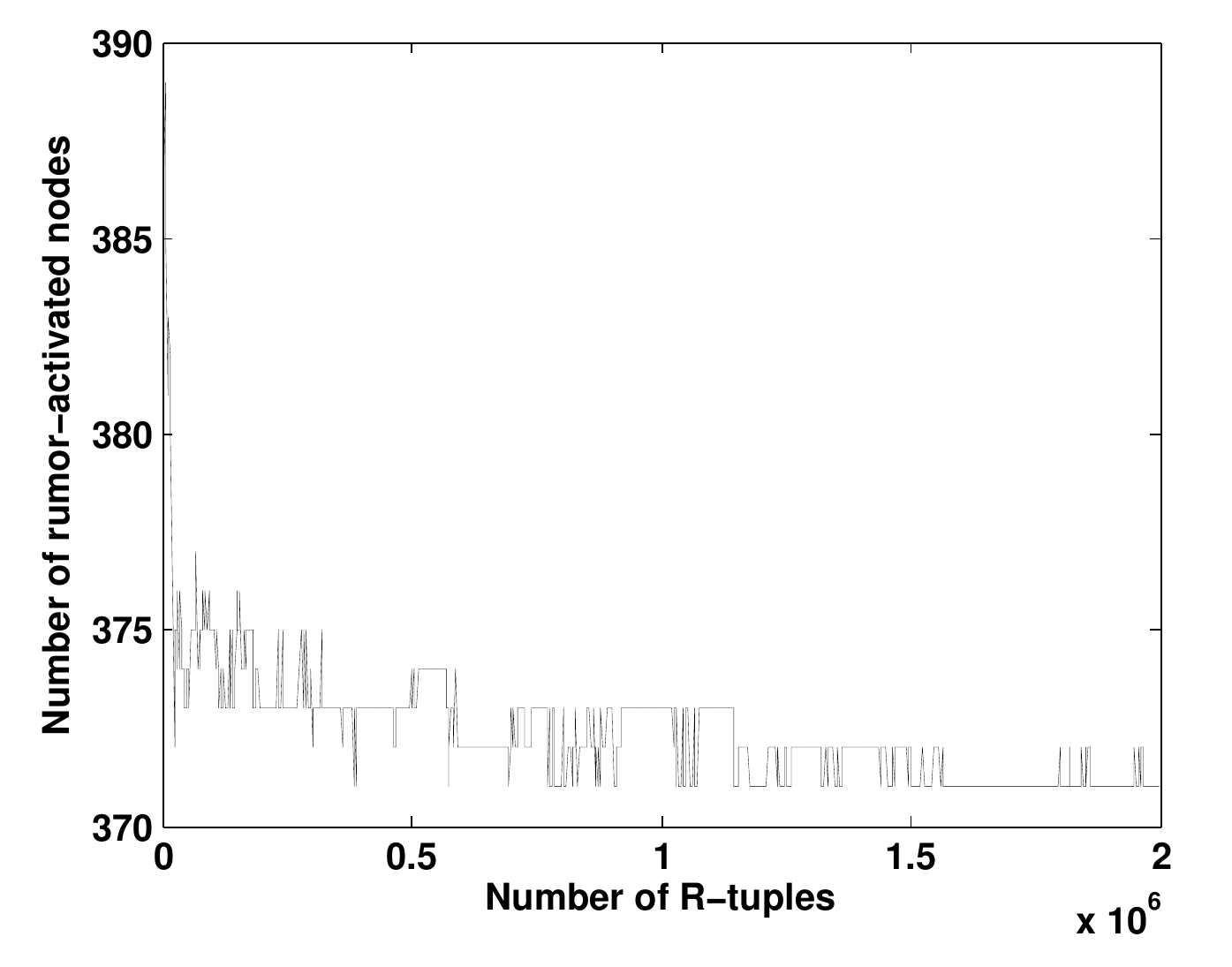}} \hspace{0mm}

\subfloat[Epinions under CP model.]{\label{fig:epin_01_r}\includegraphics[width=0.23\textwidth]{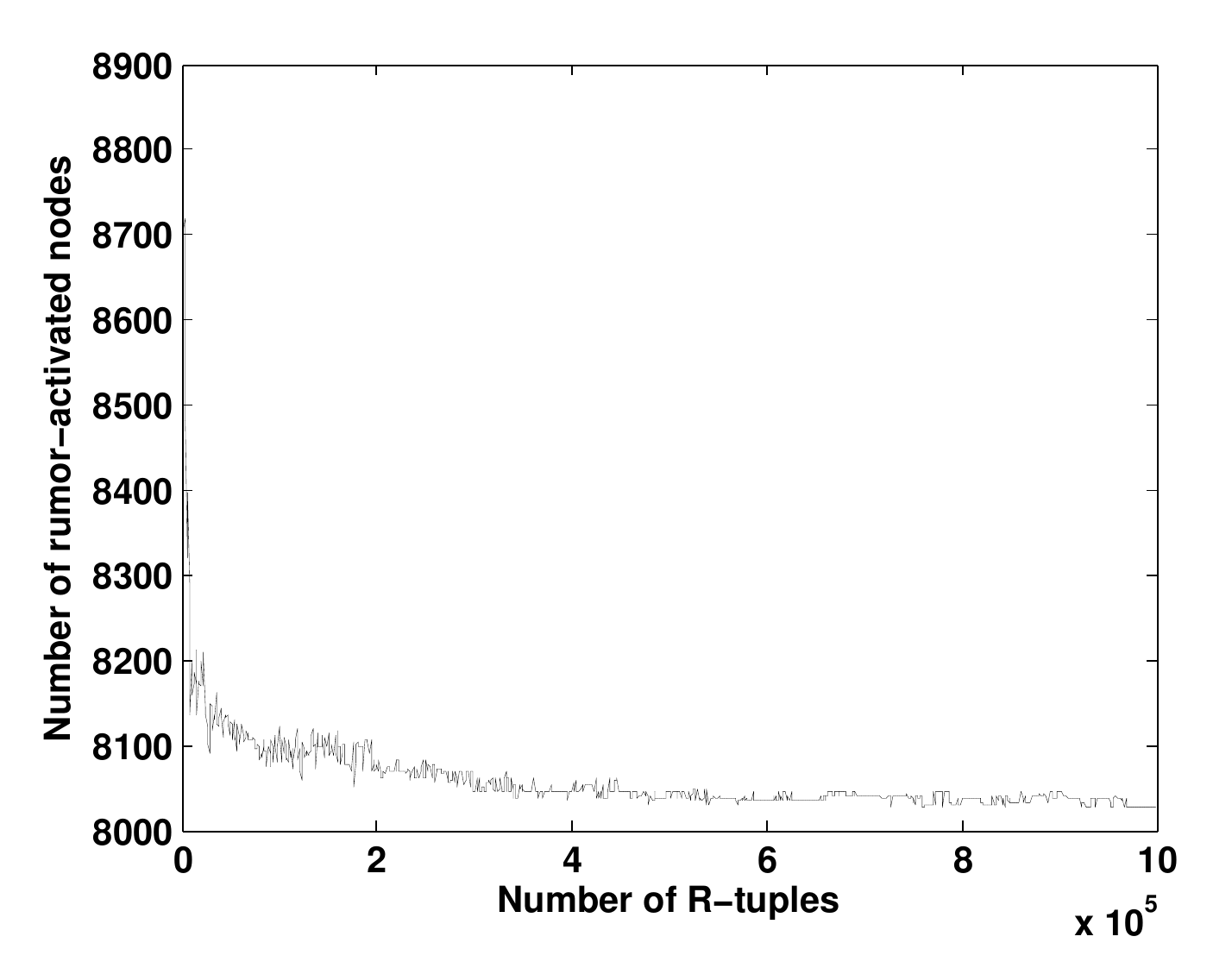}} \hspace{0mm}
\subfloat[Epinions under WC model.]{\label{fig:epin_wc_r}\includegraphics[width=0.23\textwidth]{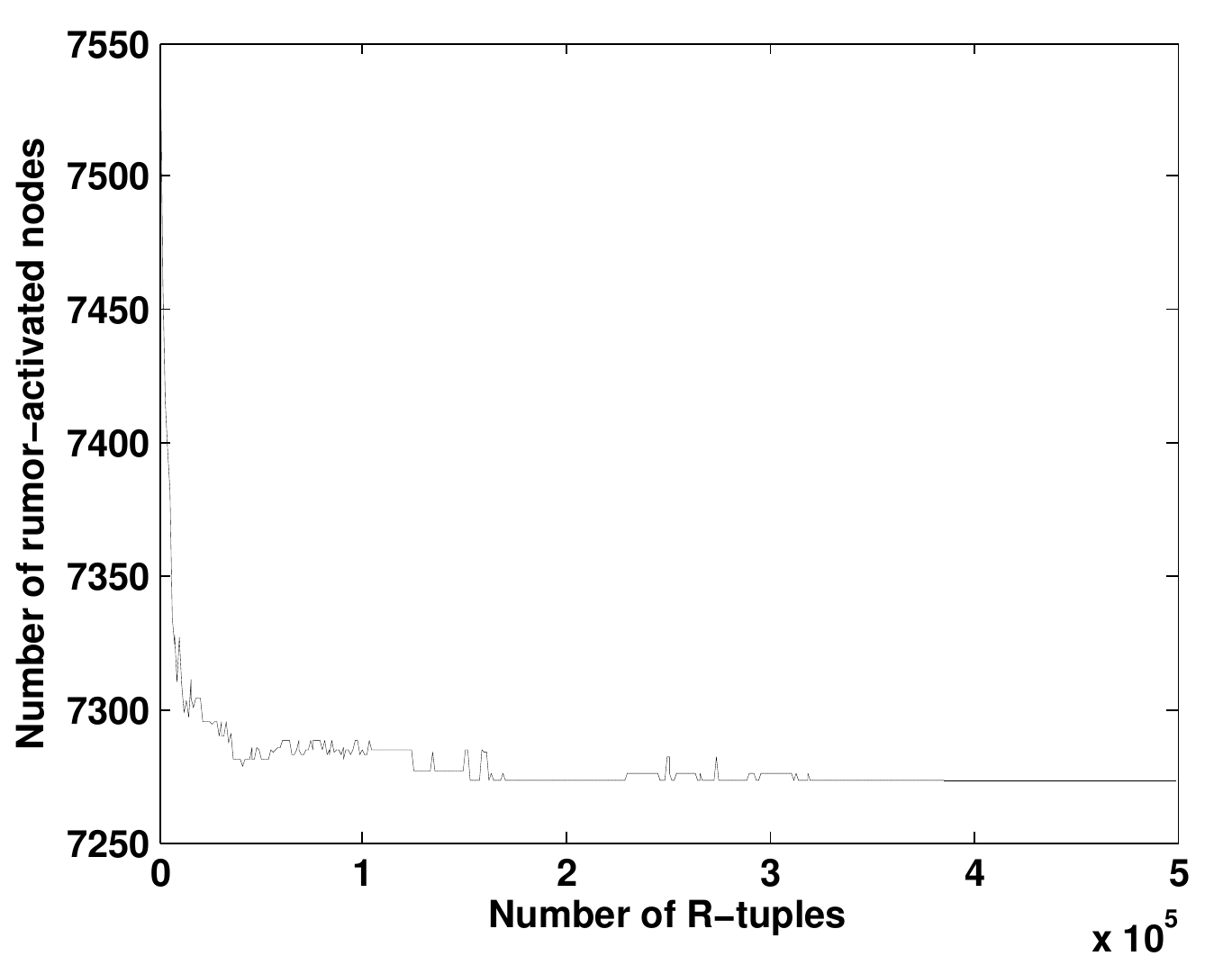}} \hspace{0mm}

\subfloat[Youtube under CP model.]{\label{fig:youtube_01_r}\includegraphics[width=0.23\textwidth]{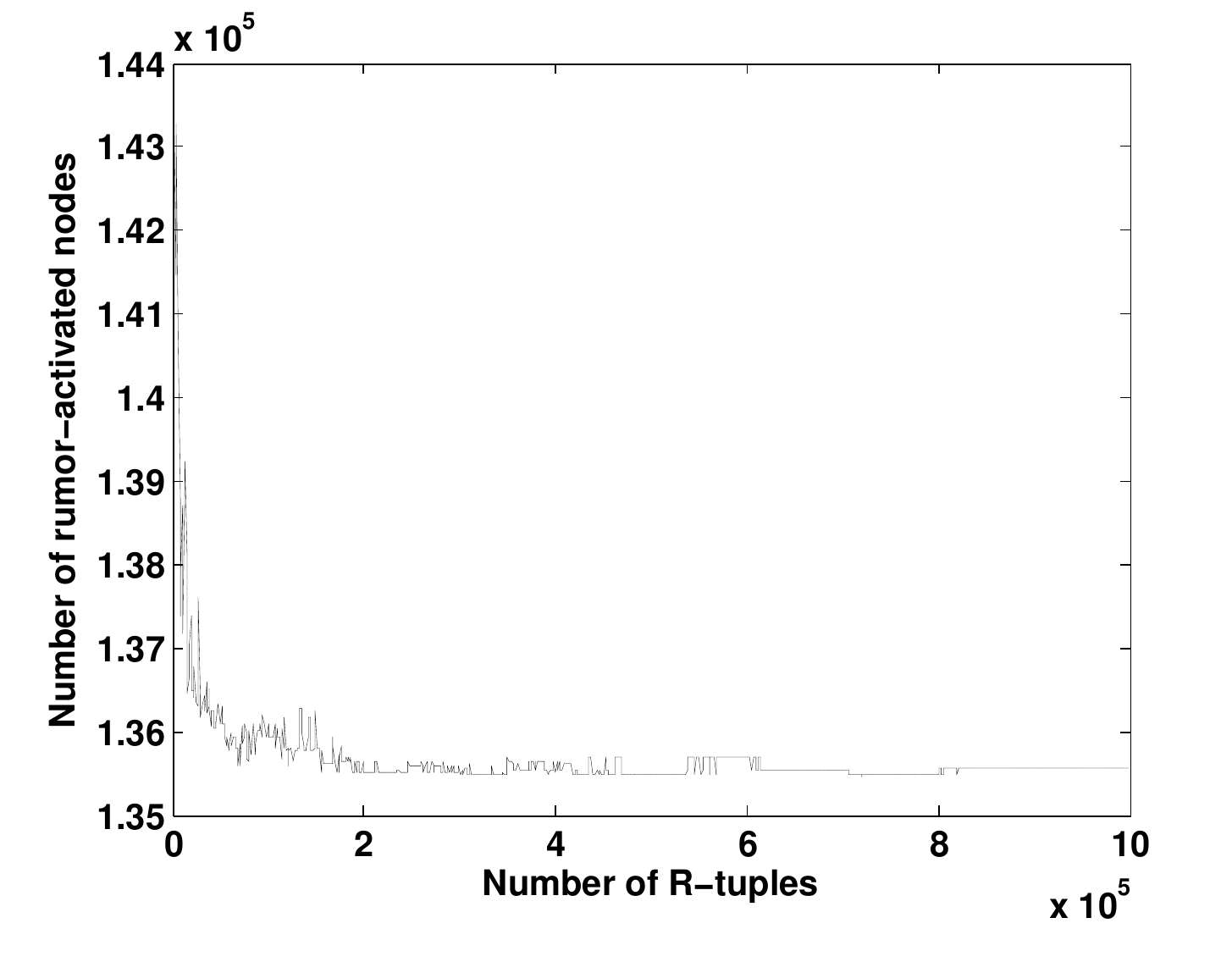}} \hspace{0mm}
\subfloat[Youtube under WC model.]{\label{fig:youtube_wc_r}\includegraphics[width=0.23\textwidth]{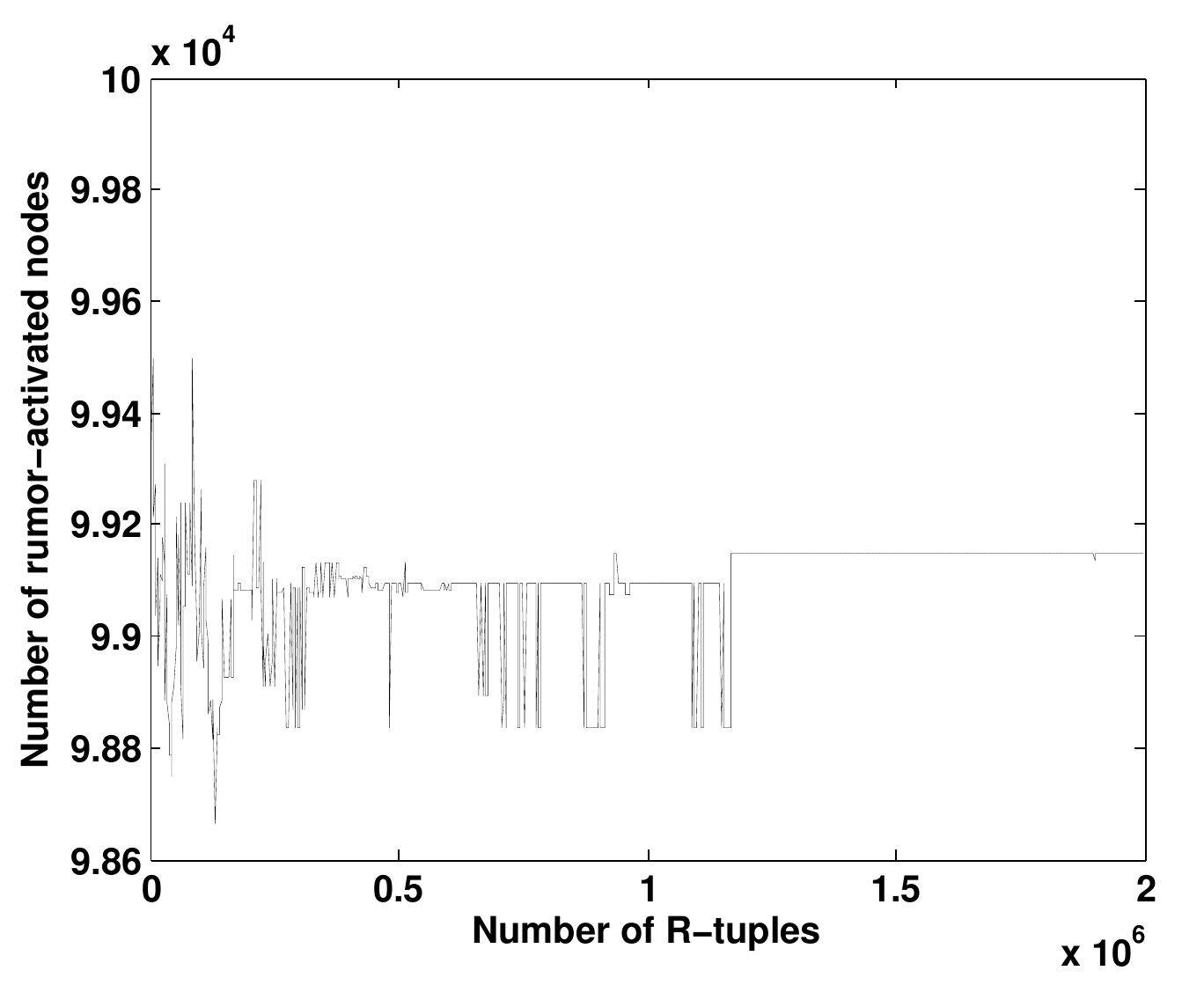}} \hspace{0mm}
\vspace{-2mm} \normalsize
\caption{\small \textbf{Results of experiment \Romannum{2}}} 
\label{fig:r}
\end{figure}

\subsubsection{Experiment \Romannum{2}}
As shown in Sec. \ref{sec:algorithm}, the main part of the analysis focuses on determining an threshold of the number of R-tuples used in Alg. \ref{alg:algortihm}. In this section, we experimentally test how the quality of the seed sets varies with the increase of used R-tuples. In particular, we are interested in that whether or not the number of R-tuples used by RBR algorithm is proper. To this end, instead of calculating $l^*$ as shown from line 2 to line 5 in Alg. \ref{alg:algortihm}, we explicitly set $l^*$ and then run the rest of the Alg. \ref{alg:algortihm} from line 6 to 8. For each dataset, we increase $l^*$ until the quality of the produced seed set tends to converge. The results are given by Fig. \ref{fig:r}. 

According to Fig. \ref{fig:r} and Table \ref{table: exp1}, RBR generates sufficient number of R-tuples in practice for most of the considered datasets. For example, on graph Power2500 under CP model, RBR totally generates 220K R-tuples as shown in the second column of Table \ref{table: exp1}, and, as shown in Fig. \ref{fig:pl_01_r}, the quality of the seed set does not markedly increases  when more than 200K R-tuples have been used. For this case, 220K R-tuples are sufficient as spending more R-tuples cannot help improve the quality. One has the same conclusion on the other three datasets. The only exception occurs on Youtube graph under WC model. For this case, RBR utilizes totally 336K R-tuples, while as shown in Fig. \ref{fig:youtube_wc_r} the standard deviation of $f()$ is about 400 when X axis is equal to 336K, and, it completely converges after 1,000K R-tuples are used. Such a case may suggest that, on large graphs, $\delta_2$ and $\delta_3$ should be set as smaller than 0.1 to raise the number of R-tuples used by Alg. \ref{alg:algortihm} such that the quality of the produced seed set can be more stable. In fact, learning the best perimeter setting is an interesting problem and we leave this part as future work.

\section{Conclusion and Future Work}
\label{sec:con}
In this paper, we have studied the rumor blocking problem for online social networks. We first design the R-tuple based sampling method and then present a randomized rumor blocking algorithm. The proposed RBR algorithm theoretically dominates the existing rumor blocking algorithms, and as shown in the experiments it is very efficient without sacrificing the blocking effect. 

One promising future work is to investigate the rumor blocking problem under other models, namely LT model. It is worthy to note that the rumor blocking problem under LT model is significantly different from that under the IC model. Another direction of future work, as mentioned in Sec. \ref{sec:exp}, is to study the parameter setting of the RBR algorithm. Finally, exact algorithm designed based on R-tuple sampling method is possibly obtainable for special graph structures like trees and regular graphs.
\appendices
\section{Proofs}
\subsection{Proof of Lemma \ref{lemma: condition}}
\label{appendix: proof_lemma: condition}
We first prove a useful property.
\begin{claim}
\label{claim: 1}
Suppose that $\dis_g(S_r,u) \neq \infty$. Let $a_u=\argmin_{v \in S_r}\dis_g(v,u)$ and $$(v_1=a_u, v_2,..., v_{l-1},v_l=u)$$ be a shortest path from $a_u$ to $u$ in $g$. If $\dis_g(S_r,u) \leq \dis_g(S_p,u)$, then all the nodes in $(v_1=a_r, v_2,...,v_l=u)$ will be activated by rumor in $g$ under $S_p$ and in particular $v_i$ will be activated at time step $i-1$.
\end{claim}
\begin{proof}
We prove this claim by induction from $v_1$ to $v_l$ along the path. First, $v_1=a_r$ is obviously activated by rumor at time step $0$ as it is a seed node of rumor. Now we prove that if $v_i$ is activated by rumor at time step $i-1$ then node $v_{i+1}$ will be activated by rumor at time step $i$. There are two cases to consider. 

\textit{Case 1.} If $v_{i+1}$ is activated by $v_i$ then clearly $v_{i+1}$ is activated by rumor at time step $i$.

\textit{Case 2.} Otherwise,  $v_{i+1}$ is activated by a neighbor $v^*$ other than $v_i$, which implies that $v^*$ is activated no later than $i-1$.  However, because $\dis_g(a_u,u)=\dis_g(S_r,u) \leq \dis_g(S_p,u)$ and  $(v_1=a_r, v_2,...,v_l=u)$ is the shortest path from $a_u$ to $u$, for any seed node $s \in S_p\cup S_r$, $\dis_g(a_u,v_i) \leq \dis_g(s, v_i)$, which means $v^*$ is activated no early than $\dis_g(a_u,v_i)=i-1$. Therefore, $v^*$ is activated at time step $i-1$ and $v^*$ and $v_i$ will attempt to activate $v_{i+1}$ simultaneously at time step $i+1$. Since rumor has the higher priority, $v_{i+1}$ will be activated by rumor regardless of whether or not $v^*$ belongs to rumor.
\end{proof}

Lemma \ref{lemma: condition} follows from the following two claims.

\begin{claim}
If $\dis_g(S_r,u) \leq \dis_g(S_p,u)$ and $\dis_g(S_r,u)\neq \infty$, then $u$ will be activated by rumor.
\end{claim}
\begin{proof}
This claim follows directly from Claim \ref{claim: 1}.
\end{proof}

\begin{claim}
If $\dis_g(S_r,u) > \dis_g(S_p,u)$ or $\dis_g(S_r,u) \neq \infty$, then $u$ will not be activated by rumor.
\end{claim}
\begin{proof}
If $\dis_g(S_r,u) = \infty$, $u$ is clearly not activated by rumor. Otherwise, let $b_u=\argmin_{v \in S_p}\dis_g(v,u)$. Similar to the proof of claim \ref{claim: 1}, we can prove that the nodes on the shortest path from $b_u$ to $u$ will activated by the positive cascade by induction. The only difference is that positive cascade always reaches those nodes earlier that rumor does by at least one time step.
\end{proof}

\subsection{Proof of Lemma \ref{lemma:key}}
\label{appendix: proof_lemma:key}
To prove Lemma \ref{lemma:key}, we introduce the following definitions.

\begin{definition} 
\label{def: x(S,T_v)}
We use $T_v=(T_v(V), T_v(E_t), T_v(E_f))$ to denote a concrete R-tuple of $v$. Let $\Re_v=\{T_v^1, T_v^2,...\}$ be the set of all possible $T_v$ and $\mathrm{Pr}[T_v]$ be the probability that $T_v$ can be generated by Alg. \ref{alg:r_tuple_v}. Given a node set $S \subseteq V$, let $x(S,T_v)$ be a variable over 0 and 1, where

$$x(S,T_v) =
  \begin{cases}
  1 &  \hspace{0mm} \hspace{-0.5mm} \text{if $S \cap T_v(V^*) \neq \emptyset $ or $T_v(B)=0$} \\
  0 & \hspace{0mm} \hspace{-0.5mm} \text{else } 
  \end{cases}$$
Different from $x(S,\T)$ defined in Def. \ref{def: x(S,T_v)}, $x(S,T_v)$ is not a random variable.
\end{definition}

\begin{definition}
\label{def:compatible}
A pair of ordered edge-sets ($E_1$, $E_2$) is \textit{valid}, if $E_1 \subseteq E$, $E_2 \subseteq E$ and $E_1 \cap E_2=\emptyset$. Note that there is a bijection between the realizations and all the valid pairs ($E_1$, $E_2$) such that $E_1 \cup E_2 = E$. We say $g$ is \textit{compatible} to ($E_1$, $E_2$) if $E_1 \subseteq E(g)$ and $E_2 \cap E(g) =\emptyset$. Let $C(E_1,E_2)$ be the set of the realizations compatible to $(E_1,E_2)$. For a valid pair ($E_1$, $E_2$) and a realization $g$ compatible to $(E_1,E_2)$, define that
\begin{equation*}
\mathrm{Pr}[(E_1,E_2)]=\prod_{e \in E_1}p_e\prod_{e \in E_2}(1-p_e),
\end{equation*}
and 
\begin{equation*}
\mathrm{Pr}[g|(E_1,E_2)]=\prod_{e \in E(g)\setminus E_1}p_e\prod_{e \notin E(g)\cup E_2}(1-p_e).
\end{equation*}
\end{definition}
One can easily check that
\begin{equation}
\label{eq:sum=1}
\sum_{g \in C(E_1, E_2)} \mathrm{Pr}[g|(E_1,E_2)]=1,
\end{equation}
and
\begin{equation}
\label{eq:condition}
\mathrm{Pr}[g]=\mathrm{Pr}[(E_1,E_2)] \cdot \mathrm{Pr}[g|(E_1,E_2)]. 
\end{equation}

Intuitively, if a realization $g$ is compatible to a valid pair $(E_1, E_2)$, $(E_1, E_2)$ can be taken as an intermediate state while generating $g$. For a R-tuple $T_v^i \in \Re_v$ of $v$, it follows that
\begin{eqnarray}
\label{eq:pr_T_v}
\Pr[T_v^i]&=&\prod_{e \in T_v^i(E_t)} p_e\prod_{e \in T_v^i(E_f)}(1-p_e)\nonumber\\
&=&\mathrm{Pr}[(T_v^i(E_t),T_v^i(E_f))]
\end{eqnarray}
Note that $((T_v^i(E_t),T_v^i(E_f))$ is always valid for each $T_v^i \in \Re_v$. Now lets consider the realizations compatible to $(T_v^i(E_t),T_v^i(E_f))$ for different R-tuples $T_v^i$ of $v$.

\begin{lemma}
\label{lemma:partition}
For each node $v$, the sets $C(T_v^i(E_t), T_v^i(E_f))$ for $T_v^i \in \Re_v$ form a partition of $\mathcal{G}$.
\end{lemma}
\begin{proof}
First, it is obvious that for each $g \in \mathcal{G}$ there exists a $T_v^i$ such that $g$ is compatible to $(T_v^i(E_t), T_v^i(E_f))$. Thus, it suffices to show that $$C(T_v^i(E_t), T_v^i(E_f)) \cap C(T_v^j(E_t), T_v^j(E_f)) = \emptyset$$ for $i \neq j$. Since $T_v^i$ and $T_v^j$  are different, there must be an edge $e$ such that $e \in T_v^i(E_t)$, $e \notin T_v^i(E_f)$, $e \notin T_v^j(E_t)$ and $e \in T_v^j(E_f)$. By Def. \ref{def:compatible}, a realization $g$ cannot be compatible to both $T_v^i$ and $T_v^j$. Lemma \ref{lemma:partition} thus proved.
\end{proof}

%Now we are ready to prove Lemma \ref{lemma:key}. According to the construction of $T$,
%$$E[x(S,\T)] = \frac{\sum_{v} \sum_{T_v \in \mathcal{T}_v}\mathrm{Pr}[T_v] \cdot x(S,T_v)}{n}$$. 

The following lemma establishes the relationship between realization and R-tuple.
\begin{lemma}
\label{lemma: T_v_g}
For a R-tuple $T_v$ of $v$, a realization $g$ compatible to $(T_v(E_t),T_v(E_f))$ and a node set $S$ of positive cascade, $v$ is not activated by rumor in $g$ if and only if $x(S, T_v)=1$.
\end{lemma}
\begin{proof}
This lemma is intuitive and it follows directly from Lemma \ref{lemma: condition}. According to Lemma \ref{lemma: condition} and Def. \ref{def: x(S,T_v)}, it suffices to show that $S \cap T_v(V^*) = \emptyset $ and $T_v(B)=1$ if and only if $\dis_g(S_r,v) \leq \dis_g(S,v)$ and $\dis_g(S_r,v)\neq \infty$. 

First, it is clear that $T_v(B)=1$ if and only if $\dis_g(S_r,v)\neq \infty$. This is because $g$ is compatible to $(T_v(E_t),T_v(E_f))$ and, by Alg. \ref{alg:r_tuple_v}, $v$ can reach some rumor seed node in $g$ if and only if $T_v(B)=1$.

Now suppose that $T_v(B)=1$ and let $a$ be one of the nodes in $V_1 \cap S_r$ tested in line 18 of Alg. \ref{alg:r_tuple_v}. Because Alg. \ref{alg:r_tuple_v} goes with a breadth-first search, $a=\dis_g(S_r,v)$ and $T_v(V^*)$ exactly consists of the node(s) $u$ such that $\dis_g(u,v)<\dis_g(S_r,v)$. Therefore,  $S \cap T_v(V^*) = \emptyset $ if and only if $\dis_g(S_r,v) \leq \dis_g(S,v)$.
\end{proof}
Now we are ready to prove Lemma \ref{lemma:key}.
 
\begin{proof}
By Defs. \ref{def: x(S,T)} and \ref{def: x(S,T_v)},
$$E[x(S,\T)] 
= \frac{\sum_{v} \sum_{T_v \in \Re_v}\mathrm{Pr}[T_v] \cdot x(S,T_v)}{n}.$$ 
According to Eq. (\ref{eq:f(S_p)}), to prove Lemma \ref{lemma:key}, it suffices to prove that
$$\sum_{T_v \in \Re_v} \mathrm{Pr}[T_v] \cdot x(S,T_v)=\sum_{g \in \mathcal{G}} \mathrm{Pr}[g] \cdot  f_g(S,v).$$
Since $(T_v(E_t), T_v(E_f))$ is a valid pair, by Eq. (\ref{eq:sum=1}),
\begin{eqnarray*}
&&\sum_{T_v \in \Re_v} \mathrm{Pr}[T_v] \cdot x(S,T_v)\\
&=& \sum_{T_v \in \Re_v} \mathrm{Pr}[T_v] \cdot \\
&&\sum_{g \in C(T_v(E_t), T_v(E_f))} \mathrm{Pr}[g|(T_v(E_t), T_v(E_f))] \cdot  x(S,T_v)
\end{eqnarray*}
Substituting $\Pr[T_v]$ with Eq. (\ref{eq:pr_T_v}) yields that

\begin{eqnarray*}
&&\sum_{T_v \in \Re_v} \mathrm{Pr}[T_v] \cdot x(S,T_v)\\
&=& \sum_{T_v \in \Re_v} \sum_{g \in C(T_v(E_t), T_v(E_f))} \\
&&\Pr[(T_v(E_t),T_v(E_f))] \cdot  \mathrm{Pr}[g|(T_v(E_t), T_v(E_f))] \cdot  x(S,T_v)
\end{eqnarray*}
According to Eq. (\ref{eq:condition}), 
$$\Pr[(T_v(E_t),T_v(E_f))] \cdot  \mathrm{Pr}[g|(T_v(E_t), T_v(E_f))]=\Pr[g],$$
and therefore,

\begin{eqnarray*}
&&\sum_{T_v \in \Re_v} \mathrm{Pr}[T_v] \cdot x(S,T_v)\\
&=& \sum_{T_v \in \Re_v} \sum_{g \in C(T_v(E_t), T_v(E_f))}\Pr[g] \cdot  x(S,T_v)\\
&&\{\text{By Lemma \ref{lemma: T_v_g}}\} \\
&=& \sum_{T_v \in \Re_v} \sum_{g \in C(T_v(E_t), T_v(E_f))}\Pr[g] \cdot  f_g(S,v)
\end{eqnarray*}

By Lemma \ref{lemma:partition}, $C(T_v(E_t), T_v(E_f))$ for $T_v \in \Re_v$ forms a partition of $\mathcal{G}$, and as a result, 
 
\begin{eqnarray*}
&& \sum_{T_v \in \Re_v} \sum_{g \in C(T_v(E_t), T_v(E_f))}\Pr[g] \cdot  f_g(S,v)\\
&=& \sum_{g \in \mathcal{G}} \mathrm{Pr}[g] \cdot  f_g(S,v).
\end{eqnarray*}
Thus, proved.
%\begin{eqnarray*}
%&& E[x(S,\T)] \\
%&=& \frac{\sum_{v} \sum_{T_v \in \mathcal{T}_v}\mathrm{Pr}[T_v] \cdot x(S,T_v)}{n} \\
%&&\{\text{By Eq. (\ref{eq:sum=1}})\} \\
%&=& \sum_{v} \sum_{T_v \in \mathcal{T}_v}\mathrm{Pr}[T_v] \cdot \\
%&& \dfrac{(\sum_{g \in C(T_v(E_t), T_v(E_f))} \mathrm{Pr}[g|(T_v(E_t), T_v(E_f))] )\cdot  x(S,T_v)}{n} \\
%&&\{\text{By Eq.~  (\ref{eq:pr_T_v})}\} \\
%&=& \sum_{v} \sum_{T_v \in \mathcal{T}_v}\mathrm{Pr}[(T_v(E_t),T_v(E_f))] \cdot \\
%&& \dfrac{(\sum_{g \in C(T_v(E_t), T_v(E_f))} \mathrm{Pr}[g|(T_v(E_t), T_v(E_f))] )\cdot  x(S,T_v)}{n} \\
%&&\{\text{By Eq. (\ref{eq:condition}})\} \\
%&=& \frac{\sum_{v} \sum_{T_v \in \mathcal{T}_v}\sum_{g \in C(T_v(E_t), T_v(E_f))} \mathrm{Pr}[g] \cdot  x(S,T_v)}{n} 
%\end{eqnarray*}
%Note that for a realization $g$ compatible to $(T_v(E_t), T_v(E_f))$ and a node set $S$, by Lemma \ref{lemma: condition}, $v$ is not activated by rumor in $g$ under $S$ if and only if $x(S,T_v(V))=1$. Therefore,

%\begin{eqnarray*}
%&&\frac{\sum_{v} \sum_{T_v \in \mathcal{T}_v}\sum_{g \in C(T_v(E_t), T_v(E_f))} \mathrm{Pr}[g] \cdot  x(S,T_v)}{n} \\
%&=&\frac{\sum_{v} \sum_{T_v \in \mathcal{T}_v}\sum_{g \in C(T_v(E_t), T_v(E_f))} \mathrm{Pr}[g] \cdot  f_g(S,v)}{n} \\
%&&\{\text{By Lemma  \ref{lemma:partition}}\} \\
%&=& \frac{\sum_{v} \sum_{g \in \mathcal{G}} \mathrm{Pr}[g] \cdot  f_g(S,v)}{n} \\
%&&\{\text{By Eq. (\ref{eq:f(S_p)})}\} \\
%&=& \frac{f(S)}{n}
%\end{eqnarray*}
\end{proof}

\subsection{Proof of Lemma \ref{lemma: opt_k_1}}
\label{appendix: proof_lemma: opt_k_1}
To prove Lemma \ref{lemma: opt_k_1}, we first prove the following two claims. Suppose that Alg. \ref{alg:OPT_k} terminates at the $i^*$-th iteration. 

\begin{claim}
\label{claim: opt_k_1}
$OPT_k < x_{i^*}$ holds with probability at least $1-1/N$ where $x_{i^*}$ is given in line 5 of Alg. \ref{alg:OPT_k}.

\end{claim}  
\begin{proof}
It suffices to show that for the $i$-th iteration from line 6 to 11 in Algorithm \ref{alg:OPT_k}, the terminate condition holds with at most $\frac{1}{N\log
n}$ probability if $OPT_k < x_i$.
For a node-set $S$ with $|S|=k$,
\begin{eqnarray*}
&&　\mathrm{Pr}[n \cdot F(S,R)/l_i \geq (1+\delta) \cdot x_i] \\
&=& \mathrm{Pr}[F(S,R)-\frac{l_i \cdot f(S)}{n} \geq \frac{l_i \cdot f(S)}{n} \cdot (\frac{(1+\delta) \cdot x_i}{f(S)}-1)] \\
&&\{\text{By Eq. (\ref{eq:chernoff_1})}\}\\
&\leq& \exp(-\frac{-l_i \cdot \frac{f(S)}{n} \cdot (\frac{(1+\delta) \cdot x_i}{f(S)}-1)^2}{2+(\frac{(1+\delta) \cdot x_i}{f(S)}-1)}) \\
&\leq& \exp(-\frac{-l_i \frac{f(S)}{n} \cdot  (\frac{(1+\delta) \cdot x_i}{f(S)}-1)}{\frac{2}{(\frac{(1+\delta) \cdot x_i}{f(S)}-1)}+1})\\
&&\{\text{Since~} f(S) \leq OPT < x_i\}\\
&\leq& \exp(-\frac{-l_i \frac{f(S)}{n} \cdot \frac{\delta \cdot x_i}{f(S)}}{\frac{2}{\delta}+1})\\
&=& \exp(-\frac{-l_i \cdot \delta^2 \cdot x_i}{n \cdot (2+\delta)})= \frac{1}{N \cdot \binom{n}{k}\log n}.\\
\end{eqnarray*}
By the union bound, $n \cdot F(S,T)/l_i \geq (1+\delta_3) \cdot x_i$ holds for the $S^{‘}$ produced in line 8 Algorithm \ref{alg:OPT_k} with a probability less than $\frac{1}{N \log n}$. Therefore, the probability that Algorithm \ref{alg:OPT_k} terminates is at most $\frac{1}{N\log
n}$ when $i \leq \log i/OPT_k$. By the union bound again, the probability that Algorithm \ref{alg:OPT_k} terminates before $i^*=\log n/OPT_k$ is less than $\frac{i^*}{N \log n}\leq \frac{1}{N}$. Thus, proved.
\end{proof}

\begin{claim}
\label{claim: opt_k_2}
If $OPT_k \geq x_{i^*}$, then $OPT_k^* \leq OPT_k$ holds with a probability at least $1-1/N$.
\end{claim}

\begin{proof}
For any node-set $S$ with $|S|=k$,
\begin{eqnarray*}
&&\Pr[\frac{n \cdot F(S,R)}{l_{i^*} \cdot (1+\delta)} \geq OPT_k]\\
&=&\mathrm{Pr}[n \cdot F(S,R) \geq l_{i^*} \cdot (1+\delta) \cdot OPT_k].
\end{eqnarray*}
Since $OPT_k \geq f(S)$,
\begin{eqnarray*}
&&\mathrm{Pr}[n \cdot F(S,R) \geq l_{i^*} \cdot (1+\delta) \cdot OPT_k]\\
&\leq&\mathrm{Pr}[n\cdot F(S,R) \geq l_{i^*} \cdot f(S)+l_{i^*} \cdot \delta \cdot OPT_k]\\
&=&\mathrm{Pr}[F(S,R)-\frac{l_{i^*} \cdot f(S)}{n} \geq \frac{l_{i^*} \cdot f(S)}{n} \cdot \frac{OPT_k \cdot \delta}{f(S)}].
\end{eqnarray*}
Applying the Chernoff bound Eq. (\ref{eq:chernoff_1}), we have

\begin{eqnarray*}
&&\mathrm{Pr}[F(S,R)-\frac{l_{i^*} \cdot f(S)}{n} \geq \frac{l_{i^*} \cdot f(S)}{n} \cdot \frac{OPT_k \cdot \delta}{f(S)}]\\
&\leq& \exp(-\frac{l_{i^*} \cdot \frac{f(S)}{n} \cdot (\frac{OPT_k \cdot \delta}{f(S)})^2}{2+\frac{OPT_k \cdot \delta}{f(S)}})\\
&=& \exp(-\frac{l_{i^*} \cdot (OPT_k \cdot \delta)^2}{n \cdot (2 \cdot f(S)+OPT_k \cdot \delta)}).
\end{eqnarray*}
Again, since $OPT_k \geq f(S)$, 

\begin{eqnarray*}
&=& \exp(-\frac{l_{i^*} \cdot (OPT_k \cdot \delta)^2}{n \cdot (2 \cdot f(S)+OPT_k \cdot \delta)})\\
&\leq& \exp(-\frac{l_{i^*} \cdot (OPT_k \cdot \delta)^2}{n \cdot (2 \cdot OPT_k+OPT_k \cdot \delta)})\\
&=& \exp(-\frac{l_{i^*} \cdot OPT_k \cdot \delta^2}{n \cdot (2+\delta)}).
\end{eqnarray*}
Substituting $l_{i^*}$ with the value given in line 5 of Alg. \ref{alg:OPT_k}, one has

\begin{eqnarray*}
&&\exp(-\frac{l_{i^*} \cdot OPT_k \cdot \delta^2}{n \cdot (2+\delta)})\\
&=& \exp(-\frac{OPT_k}{x_{i^*}} \cdot \ln(\log n \cdot \binom{n}{k} \cdot N)).
\end{eqnarray*}
Because $x_{i^*} \leq OPT_k$, the above probability is no larger than $1/(\log n \cdot \binom{n}{k} \cdot N)$. By the union bound, the probability that $OPT_k^*$ (i.e. $\frac{n \cdot F(S^{'},R)}{l_{i^*} \cdot (1+\delta)}$) is larger than or equal to $OPT_k$ is no larger than $1/N$.

%\begin{eqnarray*}
%&&\mathrm{Pr}[\frac{n \cdot F(S,T)}{l_i \cdot (1+\delta)} \geq OPT_k]\\
%&=&\mathrm{Pr}[n \cdot F(S,T) \geq l_i \cdot (1+\delta) \cdot OPT_k]\\
%&\leq&\mathrm{Pr}[n\cdot F(S,T) \geq l_i \cdot f(S)+l_i \cdot \delta \cdot OPT_k]\\
%&\leq&\mathrm{Pr}[F(S,T)-\frac{l_i \cdot f(S)}{n} \geq \frac{l_i \cdot f(S)}{n} \cdot \frac{OPT_k \cdot \delta}{f(S)}]\\
%&&\{\text{By Eq. (\ref{eq:chernoff_1})}\}\\
%&\leq& \exp(-\frac{l_i \cdot \frac{f(S)}{n} \cdot (\frac{OPT_k \cdot \delta}{f(S)})^2}{2+\frac{OPT_k \cdot \delta}{f(S)}})\\
%&=& \exp(-\frac{l_i \cdot (OPT_k \cdot \delta)^2}{n \cdot (2 \cdot f(S)+OPT_k \cdot \delta)})\\
%&\leq& \exp(-\frac{l_i \cdot (OPT_k \cdot \delta)^2}{n \cdot (2 \cdot OPT_k+OPT_k \cdot \delta)})\\
%&=& \exp(-\frac{l_i \cdot OPT_k \cdot \delta^2}{n \cdot (2+\delta)})\\
%&=& \exp(-\frac{OPT_k}{x} \cdot \ln(\log n \cdot \binom{n}{k} \cdot N))\\
%&\leq& 1/(\log n \cdot \binom{n}{k} \cdot N).
%\end{eqnarray*}

\end{proof}
Lemma \ref{lemma: opt_k_1} follows from the above two claims immediately.

\subsection{Proof of Lemma \ref{lemma: opt_k_2}}
\label{appendix: proof_lemma: opt_k_2}
The following claim helps prove Lemma \ref{lemma: opt_k_2}.
\begin{claim}
\label{claim: 3}
Let $S^{'}$ and $R$ be the sets used in the i-th iteration of Alg. \ref{alg:OPT_k}. If  $OPT_k \geq \frac{(1+\delta)^2}{1-1/e}x_i$, then $n \cdot F(S^{'},T)/l_i \leq (1+\delta) \cdot x$ holds with at most $1/N$ probability.
\end{claim}
\begin{proof}
For any seed set $S$,
\begin{eqnarray*}
&&\mathrm{Pr}[n \cdot F(S,R)/l_i \leq (1+\delta) \cdot x_i]\\
&\leq& \mathrm{Pr}[n(1-1/e)F(S_k,R)/l_i \leq (1+\delta)x]\\
&\leq& \mathrm{Pr}[n \cdot F(S_k,R)/l_i \leq \frac{OPT_k}{1+\delta}]\\
&=& \mathrm{Pr}[F(S_k,R)-\frac{l_i \cdot f(S_k)}{n} \leq \frac{l_i \cdot f(S_k)}{n}  \cdot \frac{-\delta}{1+\delta}]\\
&&\{\text{By Eq. (\ref{eq:chernoff_2})}\}\\
&\leq& \exp(-\frac{l_i \cdot f(S_k) \cdot (\frac{-\delta}{1+\delta})^2}{2})\\
&\leq& \exp(-\frac{l_i \cdot \frac{(1+\delta)^2}{1-1/e} \cdot x_i \cdot (\frac{-\delta}{1+\delta})^2}{2})\\
&\leq& 1/(N \cdot \binom{n}{k}).
\end{eqnarray*}
By the union bound, the above holds for the $S^{'}$ produced in line 8 Algorithm \ref{alg:OPT_k} with a probability less than $1/N$.
\end{proof}

Now we are ready to prove Lemma \ref{lemma: opt_k_2}. Suppose $n/2^{i+1} \leq \frac{OPT_k \cdot (1-1/e)}{(1+\delta)^2} \leq n/2^i$ for some $i$. If Algorithm \ref{alg:OPT_k} terminates before the $(i+1)$-th iteration, then, by line 9 in Algorithm \ref{alg:OPT_k}, $OPT_k^* \geq n/2^i \geq \frac{OPT_k \cdot (1-1/e)}{(1+\delta)^2}$. Now suppose Algorithm \ref{alg:OPT_k} enters the $(i+1)$-th iteration. By Claim \ref{claim: 3}, it will terminate with probability at least $1-1/N$, which means $OPT_k^* \geq n/2^{i+1} \geq \frac{OPT_k \cdot (1-1/e)}{2 \cdot (1+\delta)^2}$.

\subsection{Proof of Lemma \ref{lemma:1}}
\label{appendix: proof_lemma:1}
\begin{proof}
Recall that $E[x(S_k,\T)]=OPT_k/n$.

\begin{eqnarray*}
&&\mathrm{Pr}[\frac{n}{l}\cdot F(S_k,R_l) \leq (1-\delta_1) \cdot OPT_k] \\
&=&\mathrm{Pr}[F(S_k,R_l) \leq \frac{l}{n} (1-\delta_1) \cdot OPT_k] \\
&=&\mathrm{Pr}[F(S_k,R_l)- \frac{l}{n} \cdot OPT_k \leq -\frac{l}{n} \cdot \delta_1 \cdot OPT_k] \\
&&\{\text{By Eq. (\ref{eq:chernoff_2})\}} \\
&\leq& \exp(-\frac{l \cdot \frac{OPT_k}{n} \cdot \delta_1^2}{2}) 
\end{eqnarray*}
Because $l \geq l_1$ and Eq. (\ref{eq: opt_k^*}) holds, the above probability is no larger than $1/N$. As a result, $\frac{n}{l}\cdot F(S_k,R_l) \geq (1-\delta_1) \cdot OPT_k$
holds with probability at least $1-1/N$.

\end{proof}
\subsection{Proof of Lemma \ref{lemma:2}}
\label{appendix: proof_lemma:2}
\begin{proof}
Let $\delta_*=\delta_2-(1-1/e) \cdot \delta_1$. For any fixed $S \subseteq V$ where $|S|=k$, it follows that
\begin{eqnarray}
&&\mathrm{Pr}[F(S,R_l)-\frac{l}{n} \cdot f(S) \geq (\delta_2-(1-1/e) \cdot \delta_1) \cdot \frac{l}{n} \cdot OPT_k] \nonumber\\
&=&\mathrm{Pr}[F(S,R_l)-\frac{l}{n} \cdot f(S) \geq \delta_* \cdot \frac{lf(S)}{n} \cdot \frac{OPT_k}{f(S)}]\nonumber\\
&&\{\text{By Eq. (\ref{eq:chernoff_1})\}} \nonumber\\
&\leq& \exp(-\frac{l \cdot \frac{f(S)}{n} \cdot (\delta_* \cdot \frac{OPT_k}{f(S)})^2}{2+(\delta_* \cdot \frac{OPT_k}{f(S)})})\nonumber\\
&=& \exp(-\frac{l \cdot (\delta_* \cdot OPT_k)^2}{n\cdot(2 \cdot f(S)+\delta_* \cdot OPT_k)})\nonumber\\
&&\{\text{Since $f(S)\leq OPT_k$}\} \nonumber\\
&\leq& \exp(-\frac{l \cdot \delta_*^{2} \cdot OPT_k}{n \cdot (2+\delta_*)})\nonumber \\
&&\{\text{By Eqs. (\ref{eq:l_2}) and (\ref{eq: opt_k^*})\}} \nonumber\\
&\leq& 1/(N_2 \cdot \binom{n}{k})\nonumber
\end{eqnarray}
Since there are at most $\binom{n}{k}$ subsets of $V$ with $k$ elements, by the union bound, Eq. (\ref{eq: accuracy_S^'}) with probability at least $1-1/N$.
\end{proof}

\subsection{Proof of Lemma \ref{lemma:TIME}}
\label{appendix: proof_lemma:TIME}
\begin{proof}
For a R-tuple $T_v$ of node $v$, let $TIME(T_v)$ be the time consumed in generating $T_v$. Thus,
{\small
\begin{eqnarray*}
TIME=\frac{\sum_{v}\sum_{T_v \in \Re_v}\mathrm{Pr}[T_v] \cdot TIME(T_v)}{n}
\end{eqnarray*}
}
Note that  $TIME(T_v)$ is equal to the number of edges tested during the generation of $T_v$. Thus,
{
\begin{eqnarray*}
&&\frac{\sum_{v}\sum_{T_v \in \Re_v}\mathrm{Pr}[T_v] \cdot TIME(T_v)}{n}\\
&=&\frac{\sum_{v}\sum_{T_v \in \Re_v}\mathrm{Pr}[T_v] \cdot \sum_{(u^*,v^*)\in E}x(\{u^*\},T_v)}{n}\\
&=&\frac{\sum_{(u^*,v^*)\in E}\sum_{v}\sum_{T_v \in \Re_v}\mathrm{Pr}[T_v] \cdot x(\{u^*\},T_v)}{n}\\
&=&\frac{\sum_{(u^*,v^*)\in E}f(\{u^{*}\})}{n}\\
&\leq&\frac{m \cdot OPT_1}{n}
\end{eqnarray*} 
}
\end{proof}
\ifCLASSOPTIONcompsoc
  \section*{Acknowledgments}
  
\else
  \section*{Acknowledgment}
\fi

This work was supported by National Natural Science Foundation of China under 61472272.

\ifCLASSOPTIONcaptionsoff
  \newpage
\fi



%\bibliographystyle{IEEEtran}
\bibliographystyle{IEEEtran}
\bibliography{sigproc}


\begin{IEEEbiography}[{\includegraphics[width=1in,clip,keepaspectratio]{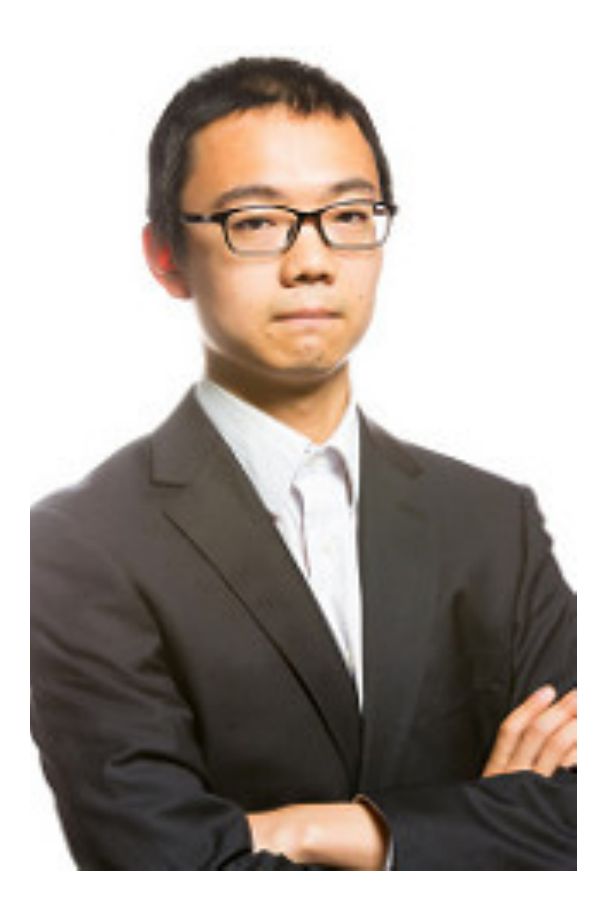}}]{Guang-mo Tong} is a Ph.D candidate in the Department of Computer Science at the University of Texas at Dallas. He received his BS degree in Mathematics and Applied Mathematics from Beijing Institute of Technology in July 2013. His research interests include computational social systems, data communication and real-time systems. He is a student member of the IEEE.
\end{IEEEbiography}

\begin{IEEEbiography}[{\includegraphics[width=1in,clip,keepaspectratio]{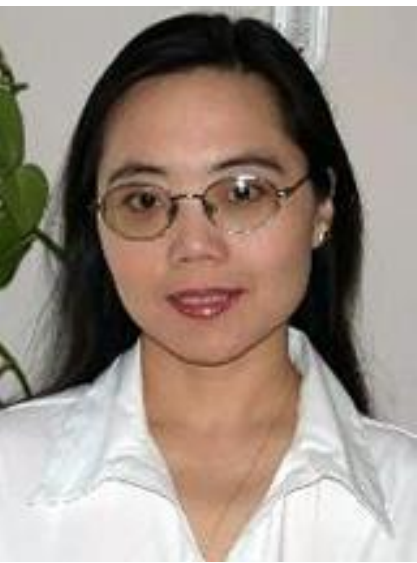}}]{Weili Wu}(M’00) is currently a Full Professor with the Department of Computer Science, University of Texas at Dallas, Dallas, TX, USA. She received the Ph.D. and M.S. degrees from the Department of Computer Science, University of Minnesota, Minneapolis, MN, USA, in 2002 and 1998, respectively. Her research mainly deals in the general research area of data communication and data management. Her research focuses on the design and analysis of algorithms for optimization problems that occur in wireless networking environments and various
database systems.
\end{IEEEbiography}

\begin{IEEEbiographynophoto}{Ling Guo} is a lecturer of the School of Information Science and Technology, Northwest University, China. He received his B.S. degree in Computer Science from Shaanxi Normal University, Xi'an, China and his PH.D degree in Renmin University of China, Beijing. From September 2015 to August 2016, he was a union Ph.D student under the supervision of Prof. Dingzhu, Du in the University of Texas at Dallas. His research interests include wireless networks and approximate algorithms. 
\end{IEEEbiographynophoto}

\begin{IEEEbiography}[{\includegraphics[width=1in,clip,keepaspectratio]{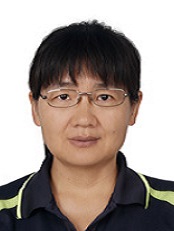}}]{Deying Li} is a professor of Renmin University of China. She received the B.S. degree and M.S. degree in Mathematics from Huazhong Normal University, China, in 1985 and 1988 respectively. She obtained the PhD degree in Computer Science from City University of Hong Kong in 2004. Her research interests include wireless networks, ad hoc and sensor networks mobile computing, distributed network system, Social Networks, and Algorithm Design etc.
\end{IEEEbiography}

\begin{IEEEbiography}[{\includegraphics[width=1in,clip,keepaspectratio]{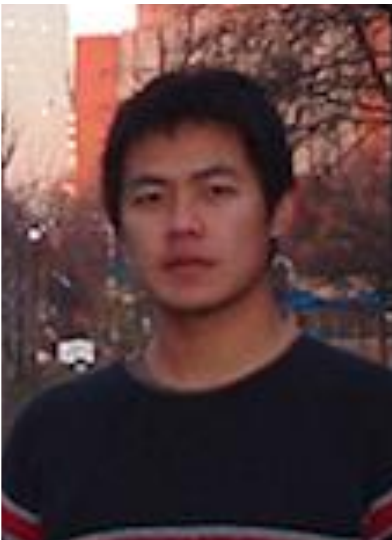}}]{Cong Liu} received the PhD degree in computer science from the University of North Carolina at Chapel Hill, in Jul. 2013. He is an assistant professor in the Department of Computer Science, University of Texas at Dallas. His research interests include real-time systems and GPGPU. He has published more than 30 papers in premier conferences and journals. He received the Best Paper Award at the 30th IEEE RTSS and the 17th RTCSA. He is a member of the IEEE.
\end{IEEEbiography}

\begin{IEEEbiography}[{\includegraphics[width=1in,clip,keepaspectratio]{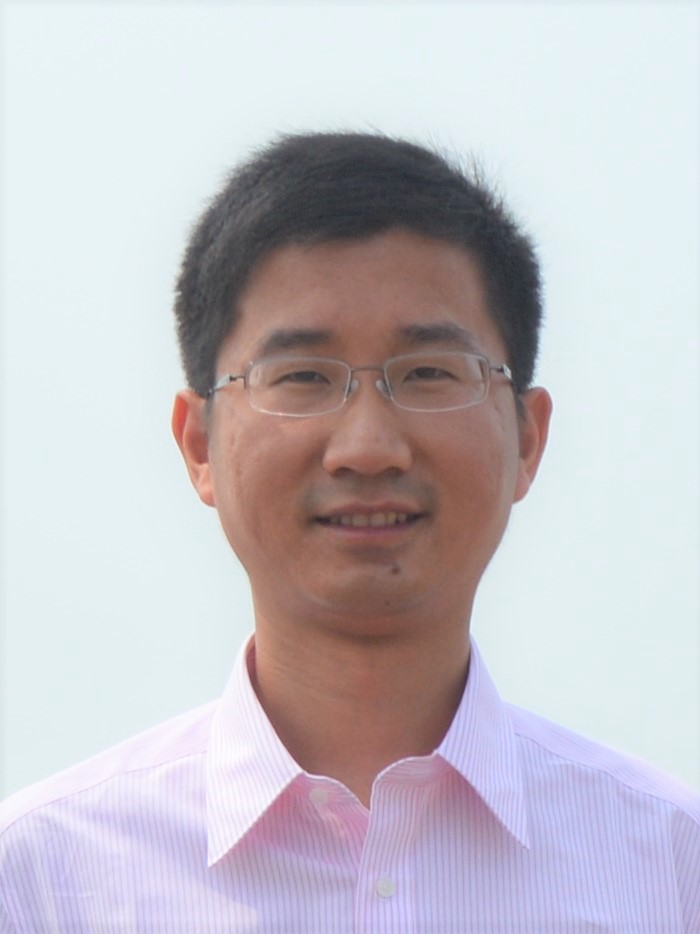}}]{Bin Liu} received the Ph.D. degree in operations research from Shandong University, China, in 2010. He joined the School of Mathematical Sciences, Ocean University of China, in 2010. His research interests include graph theory, networks, combinatorial optimization, approximation algorithm, etc.
\end{IEEEbiography}

\begin{IEEEbiography}[{\includegraphics[width=1in,clip,keepaspectratio]{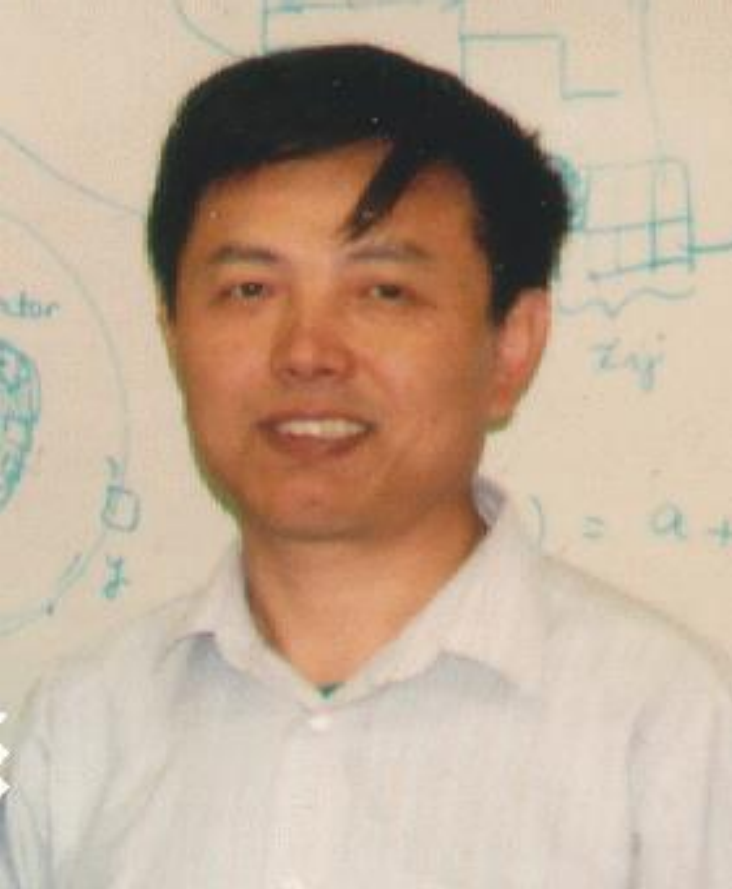}}]{Ding-Zhu Du} received the M.S. degree from the Chinese Academy of Sciences in 1982 and the Ph.D. degree from the University of California at Santa Barbara in 1985, under the supervision of Professor Ronald V. Book. Before settling at the University of Texas at Dallas, he worked as a professor in the Department of Computer Science and Engineering, University of Minnesota. He also worked at the Mathematical Sciences Research Institute, Berkeley, for one year, in the Department of Mathematics, Massachusetts Institute of Technology, for one year, and in the Department of Computer Science, Princeton University, for one and a half years. He is the editor-in-chief of the Journal of Combinatorial Optimization and is also on the editorial boards for several other journals. Forty Ph.D. students have graduated under his supervision. He is a member of the IEEE
\end{IEEEbiography}


\vfill


\end{document}